\documentclass[a4paper,11pt]{article}
\usepackage{graphicx,color}
\usepackage{amsmath,amssymb,dsfont,enumerate,fullpage,epsfig,multirow,longtable,bm}
\usepackage{epsfig,epstopdf,hhline}
\usepackage{cases}
\usepackage{algorithm}
\usepackage[noend]{algorithmic}

\usepackage[numbers]{natbib}
\usepackage[top=1 in, bottom=1 in, left=1 in, right=1 in, letterpaper]{geometry}
\usepackage{tikz}
\usetikzlibrary{arrows,positioning,decorations.pathreplacing,shapes}

\newenvironment{proof}{\noindent\emph{Proof\ }}{\hspace*{\fill}$\Box$\medskip}

\newtheorem{theorem}{Theorem}
\newtheorem{intheorem}{Informal Theorem}
\newtheorem{definition}{Definition}
\newtheorem{lemma}{Lemma}

\usepackage{amsmath}
\usepackage{paralist}
\usepackage{framed}

\newcommand\restr[2]{{% we make the whole thing an ordinary symbol
  \left.\kern-\nulldelimiterspace % automatically resize the bar with \right
  #1 % the function
  \vphantom{\big|} % pretend it's a little taller at normal size
  \right|_{#2} % this is the delimiter
  }}

\newcommand{\vect}[1]{\ensuremath{\bm{#1}}}

\newcommand{\PoA}{\text{PoA}}

\newcommand{\E}{\ensuremath{\mathbb{E}}}

\title{Game Efficiency through Linear Programming Duality}

\author{
Nguyen Kim Thang\thanks{Research supported by the ANR project OATA n\textsuperscript{o} ANR-15-CE40-0015-01}\\
\texttt{thang@ibisc.fr} \\
IBISC, University Paris-Saclay}

\date{}

\begin{document}

\maketitle

\begin{abstract}
The efficiency of a game is typically quantified by the price of anarchy (PoA), defined as the worst ratio of
the objective function value of an equilibrium --- solution of the game --- and that of an optimal outcome. 
Given the tremendous impact of tools from mathematical programming in the design of algorithms
and the similarity of the price of anarchy and different measures %in the worst-case model 
such as the approximation 
and competitive ratios, it is intriguing to develop a duality-based method to characterize 
the efficiency of games.

In the paper, we present an approach based on linear programming duality to study the efficiency of games. 
We show that the approach provides a general recipe to analyze the efficiency of games and 
also to derive concepts leading to improvements. The approach is particularly appropriate to bound the PoA. 
Specifically, in our approach the dual programs naturally lead to competitive PoA bounds that are (almost) optimal for 
several classes of games. The approach indeed captures the smoothness framework and also some current non-smooth
techniques/concepts. We show the applicability to the wide variety of games and environments, 
from congestion games to Bayesian welfare, from full-information settings to incomplete-information ones.

\end{abstract}

\thispagestyle{empty}

\newpage

\setcounter{page}{1}

\section{Introduction}
 
Algorithmic Game Theory --- a  domain at the intersection of Game Theory and Algorithms --- 
has been extensively studied in the last two decades.
The development of the domain, as well as those of many other research fields, have witnessed a common phenomenon:
interesting notions, results have been flourished at the early stage,  then deep methods, techniques
have been established at a more mature stage leading to further achievements.  
In Algorithmic Game Theory, a representative
illustration is the notion and results on the price of anarchy and the smoothness argument method \cite{Roughgarden15:Intrinsic-robustness}.          
In a game, the price of anarchy (PoA) \cite{KoutsoupiasPapadimitriou09:Worst-case-equilibria} is defined as the worst ratio between the cost of a Nash equilibrium and that of an optimal solution. 
The PoA is now considered as standard and is the most popular measure to characterize the inefficiency of Nash equilibria --- solutions of 
a game --- in the same sense of approximation ratio in Approximation Algorithms and competitive ratio in Online Algorithms. 

Mathematical programming in general and linear programming in particular are powerful tools 
in many research fields. Among others, linear programming has a tremendous impact on the design of algorithms.
Linear programming and duality play crucial and fundamental roles in several elegant methods such as primal-dual and dual-fitting 
in Approximation Algorithms \cite{WilliamsonShmoys11:The-design-of-approximation} 
and online primal-dual framework \cite{BuchbinderNaor09:The-design-of-competitive} in Online Algorithms. Given the similarity of the notions
of PoA, approximation and competitive ratios, it is intriguing and also desirable to develop a method based on duality to characterize 
the PoA of games.   
In this paper, we present and aim at developing a framework based on linear programming duality 
to study the efficiency of games.  

\subsection{A primal-dual approach} 
In high-level, the approach follows the standard primal-dual or dual-fitting techniques in approximation/online algorithms.
The approach consists of associating a game to an optimization problem and formulate an integer program corresponding to the 
problem. Next consider the linear program by relaxing the integer constraints and its dual LP.
Then given a Nash equilibrium, construct dual variables in
such a way that one can relate the dual objective to the cost of the Nash equilibrium. The PoA is then bounded
by the primal objective (essentially, the cost of the Nash equilibrium) and the dual objective (a lower bound of the optimum cost by weak duality). 
This approach has been initiated by \citet{KulkarniMirrokni15:Robust-price} 
in which full-information games with convex objective functions have been considered.

There are two crucial steps in the approach. First, by this method, the bound of PoA is at least as large as the integrality gap.
Hence, to prove optimal PoA one has to derive a formulation (of the corresponding optimization problem) 
whose the integrality gap matches to the optimal PoA. This is very similar to the issue of linear-programming-based approaches in
Approximation/Online Algorithms. The second crucial step is the construction of dual variables. The dual variables need to 
reflect the notion of Nash equilibria as well as their properties in order to relate to the cost of the equilibrium.  
Intuitively, to prove optimal bound on the PoA, the constructed dual variables must constitute an optimal dual solution. 

To overcome these obstacles, in the paper we systematically consider \emph{configuration} linear programs and a primal-dual approach.
Given a problem (game), we first consider a natural formulation of the problem. 
Then, the approach consists of introducing exponential variables and constraints to the natural formulation 
to get a configuration LP. The additional constraints we use throughout the paper 
have intuitive and simple interpretations: one constraint guarantees that the game admits exactly one outcome and 
the other constraint ensures that if a player uses a strategy then this strategy must be a component of the outcome. 
As the result, the configuration LPs significantly improve integrality gap over that of the natural formulations.

The configuration LPs have been considered in approximation algorithms and to the best of our knowledge, 
the main approach is rounding. Here, to study the efficiency of games, we consider a primal-dual approach. 
The primal-dual approach is very appropriate to  study the PoA through the mean of configuration LPs. 
In the dual program of our configuration programs, the dual constraints naturally lead to the construction of dual variables and the PoA bounds. 
Intuitively, one dual constraint corresponds exactly to the definition of Nash equilibrium and the other constraint settles the PoA bounds.

\subsection{Overview of Results}
We illustrate the potential and the wide applicability of the approach throughout various results in the contexts of complete and 
incomplete-information environments, from the settings of congestion games to welfare maximization. 
The approach allows us to unify several 
previous results and establish new ones beyond the current techniques. It is worthy to note that the analyses are simple and 
are guided by dual LP very much in the sense of primal-dual methods in designing algorithms.
Moreover,  under the lens of LP duality, 
the notion of smooth games  in both full-information settings \cite{Roughgarden15:Intrinsic-robustness} 
and incomplete-information settings \cite{Roughgarden15:The-price-of-anarchy,SyrgkanisTardos13:Composable-and-efficient},
the recent notion of no-envy learning \cite{DaskalakisSyrgkanis16:Learning-in-auctions:} 
and the new notion of dual smooth (in this paper) can be naturally derived,
which lead to the optimal bounds on the PoA of several games.

\subsubsection{Smooth Games in Full-Information Settings}
We first revisit smooth games by the primal-dual approach
and show that the primal-dual approach captures the smoothness framework \cite{Roughgarden15:Intrinsic-robustness}. 
\citet{Roughgarden15:Intrinsic-robustness} has introduced the smoothness framework, which becomes quickly a standard technique, 
and shown that every $(\lambda,\mu)$-smooth game admits the PoA at most $\lambda/(1-\mu)$. 
Through the duality approach, we show that
in terms of techniques to study the PoA for complete information settings, 
the LP duality and the smoothness framework are exactly the same thing. 
Specifically, one of the dual constraint corresponds exactly to the definition of smooth games
given in \cite{Roughgarden15:Intrinsic-robustness}.   

\begin{intheorem}
The primal-dual approach captures the smoothness framework in full-\linebreak information settings.
\end{intheorem}

\subsubsection{Congestion Games}
We consider fundamental classes of \emph{congestion games}
in which we revisit and unify results in the atomic, non-atomic congestion games and prove 
the optimal PoA bound of coarse correlated equilibria
in splittable congestion games. 
 
\paragraph{Atomic congestion games.} In this class, although the PoA bound follows the results for smooth games,
we provide another configuration formulation and a similar primal-dual approach. The purpose of this formulation is twofold.
First it shows the flexibility of the primal-dual approach. Second, it sets up the ground for an unified approach to 
other classes of the congestion games.   

\paragraph{Non-atomic congestion games.} In this class, we re-prove the optimal PoA bound \cite{RoughgardenTardos04:Bounding-the-inefficiency}. 
Along the line toward the optimal PoA bound for non-atomic congestion games, the equilibrium characterization
by a variational inequality is at the core of the analyses \cite{RoughgardenTardos04:Bounding-the-inefficiency,CorreaSchulz08:A-geometric-approach,CominettiCorrea09:The-impact-of-oligopolistic}. In our proof, we establish the optimal PoA 
directly by the mean of LP duality. By the LP duality as the unified approach, 
one can clearly observe that non-atomic setting is a version of the atomic setting in large games 
(in the sense of \cite{FeldmanImmorlica16:The-price-of-anarchy}) in which each player weight 
becomes negligible (hence, the PoA of the former tends to that of the latter). 
Besides, an advantage with LP approaches is that one can benefit from powerful techniques that have been developing for linear programming.   
Concretely, using the general framework on resource augmentation and primal-dual recently 
presented \cite{LucarelliThang16:Online-Non-preemptive}, 
we manage to recover and extend a resource augmentation result related to non-atomic setting \cite{RoughgardenTardos02:How-bad-is-selfish-routing}. 

\begin{intheorem}
In every non-atomic congestion game, for any constant $r > 0$, 
the cost of an equilibrium in a game is at most $1/r$
that of an optimal solution in a similar game in which 
the flow amount of each demand is multiplied by a factor $(1+r)$ for $r > 0$. 
\end{intheorem}

\paragraph{Splittable congestion games.}
\citet{RoughgardenSchoppmann15:Local-smoothness} has presented a \emph{local} smoothness property, 
a refinement of the smoothness framework, and proved that every $(\lambda,\mu)$-local-smooth splittable game 
admits the PoA $\lambda/(1-\mu)$. This bound is tight for a large class of scalable cost functions in splittable games 
and holds for PoA of pure, mixed, correlated equilibria. However, this bound does not hold for coarse correlated equilibria
and it remains an intriguing open question raised in \cite{RoughgardenSchoppmann15:Local-smoothness}.
Building upon the resilient ideas of non-atomic and atomic
settings, we define a property, called \emph{dual smoothness}, which is inspired by the dual constraints. 
This new notion indeed leads to the \emph{tight} PoA bound for coarse correlated equilibria
in splittable games for a large class of cost functions (the matching lower bound is given
in \cite{RoughgardenSchoppmann15:Local-smoothness} and that holds even for pure equilibria). 

\begin{definition}
A cost function $\ell: \mathbb{R}^{+} \rightarrow \mathbb{R}^{+}$ is \emph{$(\lambda,\mu)$-dual-smooth} if for every
vectors $\vect{u} = (u_{1}, \ldots, u_{n})$ and $\vect{v} = (v_{1}, \ldots, v_{n})$, 
$$
 v \ell(u) + \sum_{i=1}^{n} u_{i}(v_{i} - u_{i}) \cdot \ell'(u) \leq \lambda \cdot v\ell(v) + \mu \cdot u\ell(u)
$$  
where $u = \sum_{i=1}^{n} u_{i}$ and $v = \sum_{i=1}^{n} v_{i}$.
A splittable congestion game is \emph{$(\lambda,\mu)$-dual-smooth} if every resource $e$ in the game, 
function $\ell_{e}$ is $(\lambda,\mu)$-dual-smooth.
\end{definition}

\begin{intheorem}	
For every $(\lambda,\mu)$-dual-smooth splittable congestion game $G$, the price of anarchy of coarse correlated equilibria of $G$
is at most $\lambda/(1-\mu)$. This bound is \emph{tight} for the class of scalable cost functions.
\end{intheorem}

\subsubsection{Welfare Maximization}
%(Indeed, in the generalization for the incomplete-information settings, 
%the primal-dual approach can be shown to capture the smooth argument in this context.)

We next consider the inefficiency of Bayes-Nash equilibria in the context of welfare maximization. 

\paragraph{Smooth Auctions.}
The notion of smooth auctions in incomplete-information settings, inspired by the original 
smoothness framework \cite{Roughgarden15:Intrinsic-robustness}, has been introduced by 
\citet{Roughgarden15:The-price-of-anarchy,SyrgkanisTardos13:Composable-and-efficient}.
This powerful notion has been widely used  
to study the PoA of Bayes-Nash equilibria (see the recent survey \cite{RoughgardenSyrgkanis16:The-Price-of-Anarchy}). 
We show that the primal-dual approach captures the smoothness framework in incomplete-information settings. 
In other words, the notion of smooth auctions can be naturally derived from dual constraints in the primal-dual approach.

\begin{intheorem}
The primal-dual approach captures the smoothness framework in incom\-plete-information settings.
\end{intheorem}

\paragraph{Simultaneous Item-Bidding Auctions.} 
Many price-of-anarchy bounds in auctions are settled by smoothness-based proofs.
However, there are price-of-anarchy bounds for auctions proved via
non-smooth techniques and these techniques seem more powerful than the smoothness framework in such auctions. Representative 
examples are the simultaneous first- and second-price auctions where players' valuations are sub-additive. 
\citet{FeldmanFu13:Simultaneous-auctions} have proved that 
the PoA is constant while the smooth argument gives only logarithmic guarantees. We show that in this context, our approach is beyond
the smoothness framework and also captures the non-smooth arguments in \cite{FeldmanFu13:Simultaneous-auctions} by re-establishing their results. Specifically, a main step in our analysis --- proving the feasibility of a dual constraint --- corresponds exactly to a crucial claim in \cite{FeldmanFu13:Simultaneous-auctions}. From this point of view, the primal-dual approach helps to identify the key steps in settling the PoA bounds.

\begin{intheorem}[\cite{FeldmanFu13:Simultaneous-auctions}]
Assume that players have independent distributions over sub-additive valuations. Then, every 
Bayes-Nash equilibrium of a first-price auction and of a second price auction has expected welfare at least 1/2 and 
1/4 of the maximal welfare, respectively. 
\end{intheorem}

Subsequently, we illuminate the potential of the primal-dual approach in formulating new concepts. 
Concretely, \citet{DaskalakisSyrgkanis16:Learning-in-auctions:} have very recently introduced 
\emph{no-envy learning dynamic} --- a novel concept of learning in auctions.
No-envy dynamics have advantages over no-regret dynamics. In particular, 
no-envy outcomes maintain the approximate welfare optimality of no-regret outcomes while ensuring the 
computational tractability. Surprisingly, there is a connection between the primal-dual approach and no-envy
dynamics. Indeed, the latter can be naturally derived from the dual constraints very much in the same way as the smoothness argument is. 
We show this connection by revisiting the following theorem by the means of the primal-dual approach. 

\begin{intheorem}[\cite{DaskalakisSyrgkanis16:Learning-in-auctions:}]
Every no-envy dynamic has the average welfare at least half the expected optimal welfare. 
\end{intheorem}

\paragraph{Sequential Auctions.}
To illustrate the applicability of the primal-dual approach, we consider thereafter another format of auctions --- sequential auctions.
In a simple model of sequential auctions, items are sold one-by-one via single-item auctions. Sequential auctions has
a long and rich literature \cite{Krishna09:Auction-theory} and sequentially selling items leads to complex issues in analyzing PoA. 
\citet{LemeSyrgkanis12:Sequential-auctions,SyrgkanisTardos12:Bayesian-sequential} 
have studied sequential auctions for matching markets and matroid auctions in complete and incomplete-information settings
in which at each step, an item is sold via the first-price auctions. 
In this paper, we consider the sequential auctions for sponsored search via the second-price auctions. Informally,
auctioneer sells advertizing slots one-by-one in the non-increasing order of click-though-rates 
(from the most attractive to the least one). At each step, players submit bid for the 
currently-selling slot and the highest-bid player receives the slot and pays the second highest bid.   
In the auction, we study the PoA of perfect Bayesian equilibria and show the following improvement
over the best-known PoA bound of 2.927 \cite{CaragiannisKaklamanis15:Bounding-the-inefficiency} for the sponsored search problem.

\begin{intheorem}
The PoA of sequential second-price auctions for the sponsored search problem is at most 2.   
\end{intheorem}
An observation is that although the behaviour of players in sequential auctions might be complex, the performance guarantee 
is better than the currently best-known one for simultaneous second price auctions for the sponsored search problem.
Consequently, this result shows that the efficiency of sequential auctions is not necessarily worse 
than the simultaneous ones (and also analyzing sequential auction is not necessarily harder 
than analyzing simultaneous ones). Moreover, using the primal-dual approach, the proof is fairly simpler than the 
smoothness-based one. 

Building upon the resilient ideas for the sponsored search problem, 
we provide an improved PoA bound of 2 for the matching market problem for which the best known PoA bound is 
$2e/(e-1) \approx 3.16$ due to \citet{SyrgkanisTardos12:Bayesian-sequential}. That also answers an question 
raised in \cite{SyrgkanisTardos12:Bayesian-sequential} whether the PoA in the incomplete-information settings must be 
strictly larger than the best-known PoA bound (which is 2) in the full-information settings.

\begin{intheorem}
The PoA of sequential first-price auctions for the matching market problem is at most 2.   
\end{intheorem}

\subsection{Related works}
As the main point of the paper is to emphasize the primal-dual approach to study game efficiency, in this section
we mostly concentrate on currently existing methods. Concrete related results will be summarized 
in the corresponding sections.   

The most closely related to our work is a recent result of \cite{KulkarniMirrokni15:Robust-price}. In their approach, 
\citet{KulkarniMirrokni15:Robust-price} considered 
a convex formulation of a given game and its dual program based on Fenchel duality. Then, given a Nash 
equilibrium, the dual variables are constructed by relating the cost of the Nash equilibrium to 
that of the dual objective. In high-level, our approach has the same idea in \cite{KulkarniMirrokni15:Robust-price} and both approaches
indeed have inspired by the standard primal-dual and dual-fitting in the design of algorithms. 
Our approach is distinguished to that in \cite{KulkarniMirrokni15:Robust-price} in the following two aspects. 
First, we consider arbitrary (non-decreasing) 
objective functions and make use of configuration LPs in order to reduce substantially the integrality gap 
while the approach in \cite{KulkarniMirrokni15:Robust-price} needs convex objective functions. 
In term of approaches based on mathematical programs in approximation algorithms, 
we have come up with stronger formulations than those in \cite{KulkarniMirrokni15:Robust-price} --- a crucial point toward optimal bounds.
Second, we have shown a wide applicability of our approach from complete to incomplete-information environments while 
the approach has been proved to be powerful in the context of complete information and a question has been raised in
a recent survey \cite{RoughgardenSyrgkanis16:The-Price-of-Anarchy}
whether the framework in \cite{KulkarniMirrokni15:Robust-price} could be extended to incomplete-information settings. 

The use of duality to study the PoA have been previously considered by \citet{NadavRoughgarden10:The-limits-of-smoothness:}
and \citet{Bilo12:A-unifying-tool}. Both paper follows the same approach which is different to ours. Roughly speaking, 
given a game they consider corresponding natural formulations and incorporate the equilibrium constraint directly to the primal. 
This approach surfers the integrality-gap issue when one considers pure Nash equilibria and the objectives are non-linear or non-convex.   

For the problems studied in the paper, we systematically strengthen natural LPs by the 
construction of the new configuration LPs presented in \cite{MakarychevSviridenko14:Solving-optimization}.
\citet{MakarychevSviridenko14:Solving-optimization} propose a scheme that consists of solving 
the new LPs (with exponential number of variables) and rounding the fractional solutions to integer 
ones using decoupling inequalities for optimization problems. Instead of rounding techniques, 
we consider primal-dual approaches which are very adequate to studying game efficiency.

The smoothness framework has been introduced by \citet{Roughgarden15:Intrinsic-robustness}. This simple, elegant framework gives
tight bounds for many classes of games in complete-information settings 
including the celebrated atomic congestion games (and others in \cite{Roughgarden15:Intrinsic-robustness,BhawalkarGairing14:Weighted-congestion}). 
Similar notion, local-smoothness \cite{RoughgardenSchoppmann15:Local-smoothness}, inspired
by the smooth argument has been used to study the PoA of splittable games in which players can split their flow
to arbitrarily small amounts and route the amount in different manner. The local-smoothness
is also powerful. It has been used to settle the PoA for a large class of cost functions in splittable games 
\cite{RoughgardenSchoppmann15:Local-smoothness}
and in opinion formation games \cite{BhawalkarGollapudi13:Coevolutionary-opinion}. 

The smoothness framework has been extended to incomplete-information environments by \citet{Roughgarden15:The-price-of-anarchy,SyrgkanisTardos13:Composable-and-efficient}. It has successfully  
yielded tight worst-case bounds for the equilibria of several widely-used auction formats. We strongly recommend 
the reader to a very recent survey \cite{RoughgardenSyrgkanis16:The-Price-of-Anarchy} 
for the applications of smoothness framework in incomplete-information settings.  
However, the smoothness argument has its limit in analyzing some auctions. As mentioned earlier, the most illustrative examples
are the simultaneous first and second price auctions where players' valuations are subadditive. 
\citet{FeldmanFu13:Simultaneous-auctions} have proved that 
the PoA is constant while the smooth argument gives only logarithmic guarantees. An interesting open direction, 
as raised in \cite{RoughgardenSyrgkanis16:The-Price-of-Anarchy},
is to develop new approaches beyond the smoothness framework.

Linear programming (and mathematical programming in general) has been a powerful tool in the development of 
game theory. There is a vast literature on this subject and we can only mention the most closely related to the paper. 
One of the most interesting recent treatments on the role of linear programming in game theory is the book \cite{Vohra11:Mechanism-design:}. 
\citet{Vohra11:Mechanism-design:} revisited fundamental results in mechanism design in an elegant manner by the means of 
linear programming and its duality. It is surprising to see that many results have been shaped nicely by LPs.

\subsection{Organization of Paper}
In Section~\ref{sec:smooth}, we revisit smooth games. 
In Section~\ref{sec:congestion}, we consider congestion games.
In Section~\ref{sec:welfare}, we study the problem of welfare maximization in Bayesian setting.
The models, definitions and related work are given in the beginning of each (sub-)section.

\section{Smooth Games under the Lens of Duality}	\label{sec:smooth}

In this section, we consider smooth games \cite{Roughgarden15:Intrinsic-robustness} in the point of view of configuration LPs and duality. 
In a game, each player $i$ selects a strategy $s_{i}$ from a set $\mathcal{S}_i$ for $1 \leq i \leq n$ and that forms a \emph{strategy profile} 
$\vect{s} = (s_{1}, \ldots, s_{n})$. The cost $C_{i}(\vect{s})$ of player $i$ is a function of the strategy profile $\vect{s}$ --- the chosen 
strategies of all players. A \emph{pure Nash equilibrium} is a strategy profile $\vect{s}$ such that no player can decrease 
its cost via a unilateral deviation; that is, for every player $i$ and every strategy $s'_{i} \in \mathcal{S}_i$, 
$$
C_{i}(\vect{s}) \leq C_{i}(s'_{i},\vect{s}_{-i})
$$
where $\vect{s}_{-i}$ denotes the strategies chosen by all players other than $i$ in $\vect{s}$. The notion of Nash equilibrium 
is extended to the following more general equilibrium concepts.

A \emph{mixed Nash equilibrium} \cite{Nashothers50:Equilibrium-points} 
of a game is a product distribution $\vect{\sigma} = \sigma_{1} \times \ldots \times \sigma_{n}$
where $\sigma_{i}$ is a probability distribution over the strategy set of player $i$ such that
no player can decrease its expected cost under $\vect{\sigma}$ via a unilateral deviation: 
$$
\E_{\vect{s} \sim \vect{\sigma}} [ C_{i}(\vect{s})]  \leq \E_{\vect{s}_{-i} \sim \vect{\sigma}_{-i}} [ C_{i}(s'_{i},\vect{s}_{-i}) ]
$$
for every $i$ and $s'_{i} \in \mathcal{S}_i$, where $\vect{\sigma}_{-i}$ is the product distribution of all $\sigma_{i'}$'s other than $\sigma_{i}$.

A \emph{correlated equilibrium} \cite{Aumannothers74:Subjectivity-and-correlation} of a game is a joint probability distribution $\vect{\sigma}$ over the strategy profile of the game 
such that 
$$
\E_{\vect{s} \sim \vect{\sigma}} [ C_{i}(\vect{s}) | s_{i}]  \leq \E_{\vect{s} \sim \vect{\sigma}} [ C_{i}(s'_{i},\vect{s}_{-i}) | s_{i}]
$$ 
for every $i$ and $s_{i}, s'_{i} \in \mathcal{S}_i$.

Finally, a \emph{coarse correlated equilibrium} \cite{MoulinVial78:Strategically-zero-sum} of a game is a joint probability distribution $\vect{\sigma}$ over the strategy profile of the game 
such that 
$$
\E_{\vect{s} \sim \vect{\sigma}} [ C_{i}(\vect{s})]  \leq \E_{\vect{s} \sim \vect{\sigma}} [ C_{i}(s'_{i},\vect{s}_{-i})]
$$ 
for every $i$ and $s'_{i} \in \mathcal{S}_i$.

These notions of equilibria are presented in the order from the least to the most general ones and a notion captures the previous one as 
a strict subset. %For a more detail and intuition about these notion, we refer the reader to \cite{}. 

The notion of smooth games and robust price of anarchy are given in \cite{Roughgarden15:Intrinsic-robustness}. 
A game with a joint cost objective function $C(\vect{s}) = \sum_{i=1}^{n} C_{i}(\vect{s})$ 
is \emph{$(\lambda,\mu)$-smooth} if for every two outcomes $\vect{s}$ and $\vect{s}^{*}$,
$$
\sum_{i=1}^{n} C_{i}(s^{*}_{i}, \vect{s}_{-i}) \leq \lambda \cdot C(\vect{s}^{*}) + \mu \cdot C(\vect{s})
$$  
The \emph{robust price of anarchy} of a game $G$ is  
\begin{align*}
\rho(G) := \inf \left \{ \frac{\lambda}{1-\mu} : \textnormal{ the game is $(\lambda,\mu)$-smooth where $\mu < 1$} \right\}
\end{align*}

%It is shown that the efficiency of the set of the general coarse correlated equilibria 
%is bounded by the robust PoA.

\begin{theorem}[\cite{Roughgarden15:Intrinsic-robustness}]	\label{thm:smooth-Roughgarden}
For every game $G$ with robust PoA $\rho(G)$, every coarse correlated equilibrium $\vect{\sigma}$ of $G$
and every strategy profile $\vect{s}^{*}$,
$$
\E_{\vect{s} \sim \vect{\sigma}}[C(\vect{s})] \leq \rho(G) \cdot C(\vect{s}^{*})
$$
\end{theorem}

Until the end of the section, we revisit this theorem by our primal-dual approach. 
\paragraph{Formulation.}
Given a game, we formulate the corresponding optimization problem by a configuration LP.  
Let $x_{ij}$ be variable indicating whether player $i$ chooses strategy $s_{ij} \in \mathcal{S}_i$.
Informally, a \emph{configuration} $A$ in the formulation is a strategy profile of the game.
Formally, a configuration $A$ consists of pairs $(i,j)$ such that $(i,j) \in A$ means that in configuration 
$A$, $x_{ij} = 1$. (In other words, in this configuration, player $i$ selects strategy $s_{ij} \in \mathcal{S}_i$.)
For every configuration $A$, let $z_{A}$ be a variable such that $z_{A} = 1$ if and only if 
$x_{ij} = 1$ for all $(i,j) \in A$. Intuitively, $z_{A} = 1$ if configuration $A$ is the outcome of the game.
For each configuration $A$, let $c(A)$ be the cost of the outcome (strategy profile) corresponding to 
configuration $A$. Consider the following formulation and the dual of its relaxation.

\begin{minipage}[t]{0.45\textwidth}
\begin{align*}
\min ~ \sum_{A} &c(A) z_{A} \\
\sum_{j: s_{ij} \in \mathcal{S}_i } x_{ij} &\geq 1 & &  \forall i\\
\sum_{A} z_{A} &= 1 & &  \\
\sum_{A: (i,j) \in A} z_{A}	&= x_{ij} & & \forall i,j \\
x_{ij}, z_{A} &\in \{0,1\} & & \forall i,j,A \\
\end{align*}
\end{minipage}
\qquad
\begin{minipage}[t]{0.5\textwidth}
\begin{align*}
\max \sum_{i} \alpha_{i} &+ \beta \\
\alpha_{i} &\leq \gamma_{ij}  & &  \forall i,j\\
\beta + \sum_{(i,j) \in A} \gamma_{ij}	&\leq c(A) & & \forall A \\
\alpha_{i} &\geq 0 & & \forall i \\
\end{align*}
\end{minipage}
In the formulation, the first constraint ensures that a player $i$ chooses a strategy $s_{ij} \in \mathcal{S}_i$. 
The second constraint means that there must be an outcome of the game. 
The third constraint guarantees that if a player $i$ selects some strategy $s_{ij}$ 
then the outcome configuration $A$ must contain $(i,j)$. 

\paragraph{Construction of dual variables.}
Assuming that the game is $(\lambda,\mu)$-smooth. Fix the parameters $\lambda$ and $\mu$.
Given a (arbitrary) coarse correlated equilibrium $\vect{\sigma}$, define dual variables as follows:
\begin{align*}
\alpha_{i} := \frac{1}{\lambda} \E_{\vect{s} \sim \vect{\sigma}}[C_{i}(\vect{s})], \qquad
\beta := - \frac{\mu}{\lambda} \E_{\vect{s} \sim \vect{\sigma}} [ C(\vect{s})], \qquad
\gamma_{ij} := \frac{1}{\lambda} \E_{\vect{s} \sim \vect{\sigma}} [ C_{i}(s_{ij},\vect{s}_{-i}) ].
\end{align*}
Informally, up to some constant factors depending on $\lambda$ and $\mu$, 
$\alpha_{i}$ is the cost of player $i$ in equilibrium $\vect{\sigma}$, 
$- \beta$ stands for the cost of the game in equilibrium $\vect{\sigma}$ and 
$\gamma_{ij}$ represents the cost of player $i$ if player $i$ uses
strategy $s_{ij}$ while other players $i' \neq i$ follows strategies in $\vect{\sigma}$.
We notice that $\beta$ has negative value.

\paragraph{Feasibility.} We show that the constructed dual variables form a feasible solution. 
The first constraint follows exactly the definition of (coarse correlated) equilibrium. 
The second constraint is exactly the smoothness definition.
Specifically, let $\vect{s}^{*}$ be the strategy profile corresponding to configuration $A$. 
Note that $\E_{\vect{s} \sim \vect{\sigma}} [ C_{i}(\vect{s}^{*})] = C_{i}(\vect{s}^{*})$. 
The dual constraint reads
$$
- \frac{\mu}{\lambda} \E_{\vect{s} \sim \vect{\sigma}} [ C(\vect{s})]
+ \sum_{i} \frac{1}{\lambda} \E_{\vect{s} \sim \vect{\sigma}} [ C_{i}(s^{*}_{i},\vect{s}_{-i}) ]
\leq \E_{\vect{s} \sim \vect{\sigma}} [ C_{i}(\vect{s}^{*})]
$$  
which is the definition of $(\lambda,\mu)$-smoothness by arranging the terms and removing the expectation.

\paragraph{Price of Anarchy.} By weak duality, the optimal cost among all outcomes of the problem (strategy profiles
of the game) is at least the dual objective of the constructed dual variables. Hence, in order to bound the PoA,
we will bound the ratio between the cost of an (arbitrary) equilibrium $\vect{\sigma}$ and the dual objective of the corresponding dual variables.
The cost of equilibrium $\vect{\sigma}$ is $\E_{\vect{s} \sim \vect{\sigma}} [ C(\vect{s})]$ while the dual objective of the constructed 
dual variables is 
$$
\sum_{i=1}^{n}  \frac{1}{\lambda} \E_{\vect{s} \sim \vect{\sigma}}[C_{i}(\vect{s})] 
- \frac{\mu}{\lambda} \E_{\vect{s} \sim \vect{\sigma}} [ C(\vect{s})]
= \frac{1 - \mu}{\lambda} \E_{\vect{s} \sim \vect{\sigma}} [ C(\vect{s})]
$$
Therefore, for a $(\lambda,\mu)$-smooth game, the PoA is at most $\lambda/(1-\mu)$.

\paragraph{Remark.} 
Having shown in \cite{Roughgarden15:Intrinsic-robustness}, 
Theorem \ref{thm:smooth-Roughgarden} applies also to outcome sequences generated by 
repeated play such as vanishing average regret. By the same duality approach, we can also recover this result (by setting 
dual variables related to the average cost during the play). 

\section{Congestion Games}	\label{sec:congestion}

\subsection{Atomic Congestion Games}

\paragraph{Model.} Atomic congestion games were defined by \citet{Rosenthal73:A-class-of-games}. In this section, we consider 
atomic weighted congestion games, a generalized version of the standard congestion game. In a game, we are given 
a ground set $E$ of resources, a set of $n$ players with strategy sets $\mathcal{S}_{1}, \ldots, \mathcal{S}_{n} \subseteq 2^{E}$ and weights 
$w_{1}, \ldots, w_{n}$ and a cost function $\ell_{e}: \mathbb{R}^{+} \rightarrow \mathbb{R}^{+}$ for each resource $e \in E$. 
Note that the weighted setting generalizes the standard congestion games in which $w_{i} = 1$ for all players $i$.
Given a strategy profile $\vect{s} = (s_{1}, \ldots, s_{n})$ where $s_{i} \in \mathcal{S}_i$ for each player $i$, 
we say that $w_{e}(\vect{s}) = \sum_{i: e \in s_{i}} w_{i}$ is the \emph{load} induced on $e$ by $\vect{s}$. 
The cost of a player $i$ is defined as $C_{i}(\vect{s}) =  \sum_{e: e \in s_{i}} w_{i} \cdot \ell_{e}(w_{e})$ where $w_{e}$ is the load on resource 
$e$ induced by profile $\vect{s}$. The total cost of the game in profile $\vect{s}$ is $C(\vect{s}) = \sum_{i=1}^{n} C_{i}(\vect{s}) 
= \sum_{e: e \in s_{i}} w_{e}(\vect{s}) \cdot \ell_{e}\bigl(w_{e}(\vect{s})\bigr) $.

The PoA of atomic congestion games has been a extensively studied topic in algorithmic game theory. Most notably, 
\citet{Roughgarden15:Intrinsic-robustness} proved that the smoothness argument gave tight bounds for (unweighted) atomic congestion games. 
For the weighted setting, \citet{BhawalkarGairing14:Weighted-congestion} showed that the smoothness framework also gave tight bounds 
for large classes of congestion games.  

In this section, we reprove the upper bound \cite{Roughgarden15:Intrinsic-robustness,BhawalkarGairing14:Weighted-congestion} 
on the PoA in atomic congestion games. The result is proved by the same duality approach described in 
Section \ref{sec:smooth}, but we keep representing here for the following purposes. First, we give a slightly different formulation
of the configuration LP. To establish smoothness, all current proofs are based on smooth-inequalities related to resources. 
The new formulation is given to capture the smooth-inequality notion on resources. Second, the new proof will be used later to show 
that in term of PoA, the atomic congestion games have a strong connection with non-atomic and splittable congestion games under 
the viewpoint of duality.     

We say that a cost function $\ell_{e}: \mathbb{R}^{+} \rightarrow \mathbb{R}^{+}$ for a resource $e$ 
is \emph{$(\lambda,\mu)$-resource-smooth} if for every sequences of non-negative real numbers 
$(a_{i})_{i=1}^{n}$ and $(b_{i})_{i=1}^{n}$, it holds that
$$
\sum_{i=1}^{n} \ell_{e}\biggl( \sum_{j=1}^{i} a_{j} + b_{i}\biggr) 
	\leq \lambda \cdot \ell_{e}\biggl( \sum_{i=1}^{n} b_{i} \biggr) + \mu \cdot \ell_{e}\biggl( \sum_{i=1}^{n} a_{i} \biggr) 
$$  

\begin{theorem}[\cite{Roughgarden15:Intrinsic-robustness,BhawalkarGairing14:Weighted-congestion}]	\label{thm:atomic}
Let $\mathcal{L}$ be a non-empty set of cost functions. The PoA of every coarse correlated equilibrium of every 
(weighted) atomic congestion game with cost functions $\ell_{e} \in \mathcal{L}$ is at most 
$$
\inf \biggl\{ \frac{\lambda}{1-\mu}: \ell_{e} \textnormal{ is $(\lambda,\mu)$-resource-smooth where $\mu < 1$ } \forall e \in E \biggr\}
$$
\end{theorem}
\begin{proof} \\
\textbf{Formulation.} Let $x_{ij}$ be variable indicating whether player $i$ chooses strategy $s_{ij} \in \mathcal{S}_i$.
For every resource $e$, let $z_{eT}$ be a variable such that 
$z_{eT} = 1$ if and only if every player $i \in T$ uses resource $e$, i.e., $e \in s_{i}$, and player 
$i \notin T$ does not use resource $e$. Denote $w(T) = \sum_{i \in T} w_{i}$. Consider the following integer program and its dual.
In the primal, the first constraint says that a player $i$ has to select a strategy $s_{ij} \in \mathcal{S}_i$. 
The second constraint means that a subset of players $T$ will use resource $e$. 
The third constraint guarantees that if a player $i$ chooses some strategy $s_{ij} \in \mathcal{S}_i$ containing 
resource $e$ then there must be a subset of players $T$ such that $i \in T$ and $z_{eT} = 1$. 

\begin{minipage}[t]{0.45\textwidth}
\begin{align*}
\min ~ \sum_{e} w(T) & \ell_{e}(w(T)) z_{eT} \\
\sum_{j} x_{ij} &\geq 1 & &  \forall i\\
\sum_{T} z_{eT} &= 1 & & \forall e \\
\sum_{T: i \in T} z_{eT}	&= \sum_{j: e \in s_{ij}} x_{ij} & & \forall i, e \\
x_{ij}, z_{eT} &\in \{0,1\} & & \forall i,j,e,T \\
\end{align*}
\end{minipage}
\qquad
\begin{minipage}[t]{0.5\textwidth}
\begin{align*}
\max \sum_{i} \alpha_{i} &+ \sum_{e} \beta_{e} \\
\alpha_{i} &\leq \sum_{e: e \in s_{ij}} \gamma_{i,e}  & &  \forall i,j\\
\beta_{e} + \sum_{i \in T} \gamma_{i,e}	&\leq w(T) \ell_{e}(w(T))  & & \forall e,T \\
\alpha_{i} &\geq 0 & & \forall i \\
\end{align*}
\end{minipage}

\paragraph{Dual Variables.} Fix parameters $\lambda$ and $\mu$.
Given a coarse correlated equilibrium $\vect{\sigma}$, define corresponding dual variables as follows.
\begin{align*}
\alpha_{i} = \frac{1}{\lambda} \E_{\vect{s} \sim \vect{\sigma}}[C_{i}(\vect{s})], \quad
\beta_{e} := - \frac{\mu}{\lambda} \E_{\vect{s} \sim \vect{\sigma}} \biggl[\sum_{i: e \in s_{i}} w_{i} \ell_{e}(w_{e}(\vect{s})) \biggr], \quad
\gamma_{i,e} = \frac{1}{\lambda} \E_{\vect{s} \sim \vect{\sigma}} \bigl[w_{i} \cdot \ell_{e}\bigl(w_{e}(\vect{s}_{-i}) + w_{i} \bigr) \bigr]
\end{align*}
where $w_{e}(\vect{s}_{-i}) = \sum_{i' \neq i, e \in s_{i'}} w_{i'}$.
Informally, up to some constant factors, $\alpha_{i}$ is the cost of player $i$ in equilibrium $\vect{\sigma}$, 
$-\beta_{e}$ stands for the total cost of players on resource $e$ in this equilibrium and 
$\gamma_{i,e}$ represents the cost of player $i$ on resource $e$ if player $i$ uses
strategy containing $e$ while other players $i'$ follows strategy $s_{i'}$ for all $i' \neq i$.

\paragraph{Feasibility.} By this definition of dual variables, the first dual constraint follows 
from the definition of coarse correlated equilibrium. 
The second dual constraint is satisfied due to the smoothness definition. 
Specifically, the constraint for a resource $e$ and a subset of players $T$ reads
\begin{align*}
- \E_{\vect{s} \sim \vect{\sigma}} \biggl[ \frac{\mu}{\lambda} w_{e}(\vect{s}) \ell_{e}\bigl(w_{e}(\vect{s})\bigr) \biggr]
+ \E_{\vect{s} \sim \vect{\sigma}} \biggl[ \frac{1}{\lambda} \sum_{i \in T} w_{i} \ell_{e}\bigl(w_{e}(\vect{s}_{-i}) + w_{i} \bigr) \biggr]
	\leq w(T) \ell_{e}(w(T))
\end{align*}
The inequality holds since without expectation and by linearity of expectation 
(and also $\E_{\vect{s} \sim \vect{\sigma}} \bigl[ w(T) \linebreak \cdot \ell_{e}(w(T)) \bigr] = w(T) \ell_{e}(w(T))$),
it is exactly the smoothness definition.

\paragraph{Bounding primal and dual.} The PoA is bounded by the ratio between the primal objective and the dual one. 
Note that $\sum_{i} \alpha_{i} =  \sum_{i} \frac{1}{\lambda} \E_{\vect{s} \sim \vect{\sigma}}[C_{i}(\vect{s})] = \E_{\vect{s} \sim \vect{\sigma}} \bigl[ \frac{1}{\lambda}\sum_{e} w_{e}(\vect{s}) \ell_{e}(w_{e}(\vect{s})) \bigr]$. 
Therefore, 
$$
\sum_{i} \alpha_{i} + \sum_{e} \beta_{e} = \frac{1-\mu}{\lambda} \sum_{e} w_{e}(\vect{s}) \ell_{e}(w_{e}(\vect{s}))
$$
Hence, $\PoA \leq \lambda/(1-\mu)$.
\end{proof}

%%%****************************
%%%****************************
\subsection{Nonatomic Congestion Games}
\paragraph{Model.} Non-atomic congestion games were defined by \citet{RoughgardenTardos04:Bounding-the-inefficiency}, 
motivated by the non-atomic routing games of 
\citet{Wardrop52:ROAD-PAPER.} and \citet{BeckmannMcGuire56:Studies-in-the-Economics} and the congestion games of 
\citet{Rosenthal73:A-class-of-games}. We consider a discrete version of non-atomic congestion games.
The main purpose of restricting to discrete settings is that we can use tools from linear programming.
The continuous settings can be done by considering successively finer discrete spaces.   

Fix a constant $\epsilon$ (arbitrarily small). A non-atomic congestion game consists of a ground set $E$ of resources and 
$n$ different types of players. The set of strategies of players of type $i$ is $\mathcal{S}_{i}$
 and each strategy consists of a subset of resources. Players of type $i$ are associated to an integer number $m_{i}$
 that corresponds to a total amount $w_{i} := m_{i} \cdot \epsilon$.
% To a player type $i$, a strategy $S \in \mathcal{S}_{i}$ and a resource $e \in S$, 
% we associate a positive \emph{rate of consumption} $r_{e,S}$ that defines the amount of 
% congestion contributed to resource $e$ by player of type $i$ selecting strategy $S$.
Players of type $i$ select strategies $s_{ij} \in \mathcal{S}_{i}$ and distribute amounts $f_{s_{ij}}$ ---
a non-negative \emph{multiple of $\epsilon$} --- to strategy $s_{ij}$, 
which lead to a strategy distribution $\vect{f} = (f_{s_{ij}})$ with $\sum_{s_{ij} \in \mathcal{S}_{i}} f_{s_{ij}} = w_{i} = m_{i} \epsilon$ 
for player type $i$. 
We abuse notation and let $f_{e}$ be the total amount of congestion induced on resource $e$ by the strategy distribution $\vect{f}$. 
That is, $f_{e} := \sum_{i=1}^{n} \sum_{e \in s_{ij}} f_{s_{ij}}$. 
Each resource has a non-decreasing cost function $\ell_{e} : \mathbb{R}^{+} \rightarrow \mathbb{R}^{+}$. 
With respect to a strategy distribution $\vect{f}$, players of type $i$ selecting strategy $s_{ij} \in \mathcal{S}_{i}$ incurs
a cost $C_{s_{ij}}(\vect{f}) = \sum_{e \in s_{ij}} \ell_{e}(f_{e})$.
A strategy distribution $\vect{f}$ is an \emph{pure equilibrium} if for each player type $i$ and strategy $s_{ij}, s_{ij'} \in \mathcal{S}_{i}$
with $f_{s_{ij}} >0$, 
$$
C_{s_{ij}}(\vect{f}) \leq C_{s_{ij'}}(\vect{f})
$$
The more general equilibrium concept such as mixed, correlated and coarse correlated equilibria, are defined 
similarly as in Section~\ref{sec:smooth}. 
The social cost of a strategy distribution $\vect{f}$ is 
$$
C(\vect{f}) = \sum_{i=1}^{n} \sum_{s_{ij} \in \mathcal{S}_{i}} f_{s_{ij}} \cdot C_{s_{ij}}(\vect{f}) 
= \sum_{e} f_{e} \cdot \ell_{e}\bigl(f_{e}\bigr) 
$$

For non-atomic congestion games, tight bounds on the PoA for almost all classes of cost function have been given in 
\cite{RoughgardenTardos04:Bounding-the-inefficiency}. The core of all analyses for PoA bounds is indeed the characterization of 
the unique equilibrium via a variational inequality due to \citet{BeckmannMcGuire56:Studies-in-the-Economics}. 
This argument is explained in \cite{CorreaSchulz08:A-geometric-approach,CominettiCorrea09:The-impact-of-oligopolistic}. 
Moreover, the connection between smoothness arguments and PoA bounds for non-atomic congestion games was revealed in \cite{CorreaSchulz08:A-geometric-approach}. 

\subsubsection{Efficiency of Non-Atomic Congestion Games}

In this section, we reprove the tight bound for non-atomic congestion games by the duality approach. It has been shown that 
in non-atomic congestion games all equilibria are essentially unique; specifically, all coarse correlated equilibria of a non-atomic congestion
game have the same cost \cite{BlumEven-Dar10:Routing-Without}. Hence, the robust PoA is indeed the PoA of pure 
Nash equilibrium. However, as we do not use the equilibrium characterization from \cite{BeckmannMcGuire56:Studies-in-the-Economics}, 
we will prove the PoA bound for coarse correlated equilibria. Consequently, the tight PoA bound can be proved for non-regret sequences
and short best-reponse sequences. Moreover, we avoid the standard assumptions on the cost functions: $x\ell_{e}(x)$ is convex and 
$\ell_{e}(x)$ is differentiable.  

Let $\mathcal{L}$ be a non-empty set of cost functions. The \emph{Pigou bound} $\xi(\mathcal{L})$ for 
$\mathcal{L}$ is defined as
$$
\xi(\mathcal{L}) := \sup_{\ell \in \mathcal{L}} \sup_{u,v} \frac{u \cdot \ell(u)}{ v \cdot \ell(v) + (u-v) \cdot \ell(u)}
$$

\begin{theorem}[\cite{RoughgardenTardos04:Bounding-the-inefficiency}]	\label{thm:non-atomic}
Let $\mathcal{L}$ be a set of cost functions. Then, for every splittable congestion game $G$ with 
cost functions in $\mathcal{L}$, the price of anarchy of $G$
is at most $\xi(\mathcal{L})$.
\end{theorem}
\begin{proof} \\
\textbf{Formulation.}
Denote a finite set of multiples of $\epsilon$ as $\{a_{0}, a_{1}, \ldots, a_{m}\}$ where $a_{k} = k\cdot \epsilon$
and $m = \max_{i=1}^{n} m_{i}$.   
We say that $T$ is a \emph{configuration} of a resource $e$ if $T$ consists of couples $(i,k)$ that specifies 
player type $i$ distributes an amount $a_{k} $ to a strategy $s_{ij} \in \mathcal{S}_{i}$ where 
$e \in s_{ij}$. Note that in a configuration $T$ of a resource $e$, 
there might be multiple couples $(i,k) \in T$ and $(i,k') \in T$ corresponding to players of the same type. 
It simply means that players of type $i$ distribute the amounts $a_{k}$ and $a_{k'}$ to some strategies $s_{ij}$ and 
$s_{ij'}$ respectively that contains resource $e$, i.e., $e \in s_{ij}$ and $e \in s_{ij'}$.
Intuitively, a configuration of a resource is a strategy distribution of a game restricted on the resource. 

Let $x_{ijk}$ be variable indicating whether player type $i$ distributes an amount $a_{k}$ to strategy $s_{ij} \in \mathcal{S}_{i}$.
For every resource $e$ and a configuration $T$, let $z_{eT}$ be a variable such that 
$z_{eT} = 1$ if and only if players type $i$ distributes $a_{k}$ to some strategy 
containing resource $e$ for $(i,k) \in T$. In other words, $z_{eT} = 1$ if and only if for $(i,k) \in T$, $x_{ijk} = 1$ 
for some $s_{ij} \in \mathcal{S}_i$ such that $e \in s_{ij}$. 
For a configuration $T$ of a resource $e$, let $w(T)$ be the total amount distributed by players on resource $e$
in this configuration.
Consider the following configuration integer program and its dual.

\begin{minipage}[t]{0.45\textwidth}
\begin{align*}
\min ~ \sum_{e} w(T) & \ell_{e}(w(T)) z_{eT} \\
\sum_{j,k} a_{k} x_{ijk} &\geq w_{i} & &  \forall i\\
\sum_{T} z_{eT} &= 1 & & \forall e \\
\sum_{T: (i,k) \in T} z_{eT}	&= \sum_{j: e \in s_{ij}} x_{ijk} & & \forall (i,k), e \\
x_{ijk}, z_{eT} &\in \{0,1\} & & \forall i,j,e,T \\
\end{align*}
\end{minipage}
\qquad
\begin{minipage}[t]{0.45\textwidth}
\begin{align*}
\max \sum_{i} w_{i} \alpha_{i} &+ \sum_{e} \beta_{e} \\
a_{k} \alpha_{i} &\leq \sum_{e: e \in s_{ij}} \gamma_{i,k,e}  & &  \forall i,k,j\\
\beta_{e} + \sum_{(i,k) \in T} \gamma_{i,k,e}	&\leq w(T) \ell_{e}(w(T))  & & \forall e,T \\
\alpha_{i} &\geq 0 & & \forall i \\
\end{align*}
\end{minipage}

In the primal, the first constraint ensures that players of type $i$ distribute the total amount $w_{i}$ among its strategies.
The second constraint means that a resource $e$ is always associated to a configuration (possibly empty). 
The third constraint guarantees that if player type $i$ distributes an amount $a_{k}$ to some strategy $s_{ij}$ containing 
resource $e$ then there must be a configuration $T$ such that $(i,k) \in T$ and $z_{eT} = 1$.

\paragraph{Dual Variables.} 
Given a coarse correlated equilibrium $\vect{\sigma}$, define the corresponding dual variables as follows.
\begin{align*}
\alpha_{i} &:= \E_{\vect{f} \sim \vect{\sigma}}\biggl[\sum_{e \in s_{ij}} \ell_{e}(f_{e})\biggr] \textnormal{ for some } s_{ij} \in \mathcal{S}_{i}: f_{s_{ij}} > 0, \\
\gamma_{i,k,e} &:=  \E_{\vect{f} \sim \vect{\sigma}} \bigl[ a_{k} \cdot \ell_{e}(f_{e}) \bigr], \\
\beta_{e} &:= \inf_{T} \biggl\{ w(T)\ell_{e}\bigl( w(T) \bigr) 
	- \E_{\vect{f} \sim \vect{\sigma}} \biggl[ \sum_{(i,k) \in T} a_{k} \cdot \ell_{e}(f_{e}) \biggr] \biggr\}
\end{align*}

The dual variables have similar interpretations as previous analysis. 
Variable $\alpha_{i}$ is the total cost of resources in a strategy used by player type $i$ in equilibrium $\vect{\sigma}$ and 
$\gamma_{i,k,e}$ represents an estimation of the cost of player $i$ on resource $e$ if player type $i$ distributes
an amount $a_{k}$ in some strategy containing $e$ while 
other players $i'$ follows their strategies in $\vect{\sigma}$.

\paragraph{Feasibility.} By this definition of dual variables, the first dual constraint holds 
since it is the definition of coarse correlated equilibrium. 
The second dual constraint for a resource $e$ and a configuration $T$ reads
\begin{align*}
\beta_{e} + \sum_{(i,k) \in T} \E_{\vect{f} \sim \vect{\sigma}} \bigl[ a_{k} \cdot \ell_{e}(f_{e}) \bigr]
	\leq w(T) \ell_{e}(w(T)) 
\end{align*}
This inequality follows directly from the definition of $\beta$-variables and
linearity of expectation.

\paragraph{Bounding primal and dual.} 
For each resource $e$, let $v_{e}$ be the amount in $T$ corresponding the 
infimum in the definition of $\beta_{e}$. (As we consider discrete and finite settings, 
the infimum is indeed an minimum.)
The dual objective is
\begin{align*}
\sum_{i} w_{i} \alpha_{i} + \sum_{e} \beta_{e}
&=  \E_{\vect{f} \sim \vect{\sigma}} \left[ \sum_{e} \biggl( f_{e} \ell_{e}(f_{e})  + v_{e} \ell_{e}(v_{e})
				- v_{e} \ell_{e}(f_{e}) \biggr) \right] 	
\end{align*}
where in the equalities, we use the definition of dual variables.
Note that the term $\bigl( f_{e} \ell_{e}(f_{e})  + v_{e} \ell_{e}(v_{e}) - v_{e} \ell_{e}(f_{e}) \bigr) \geq 0$ for every resource 
$e$. Specifically, since $\ell_{e}$ is non-decreasing, 
if $f_{e} \geq v_{e}$ then $f_{e} \ell_{e}(f_{e}) \geq  v_{e} \ell_{e}(f_{e})$; else
$v_{e} \ell_{e}(v_{e}) \geq  v_{e} \ell_{e}(f_{e})$.

Besides, the primal objective is $\E_{\vect{f} \sim \vect{\sigma}} \bigl[ \sum_{e} f_{e} \ell_{e}(f_{e}) \bigr]$.
Hence, the ratio between primal and dual is at most
$$
\max_{e} \frac{f_{e} \ell(f_{e})}{ v_{e} \ell_{e}(v_{e}) + (f_{e}-v_{e})\ell_{e}(f_{e})}
$$
which is bounded by $\xi(\mathcal{L})$ where $\mathcal{L}$ is the class of cost functions on resources in the game.
\end{proof}

\paragraph{Remark.} The proofs of Theorem \ref{thm:atomic} and Theorem \ref{thm:non-atomic} are essentially the same. 
By the duality approach as a unifying tool, the main difference in term of equilibrium efficiency between atomic and non-atomic congestion games 
is due to the definition of player cost. In the context of large games \cite{FeldmanImmorlica16:The-price-of-anarchy}, 
while the weight of a player is negligible then the player cost in a atomic congestion game coincides with the one in the corresponding 
non-atomic congestion game. In this context, the PoA in atomic congestion game tends to that in non-atomic setting.    

\subsubsection{Resource Augmentation in Non-Atomic Congestion Games}
\citet{RoughgardenTardos02:How-bad-is-selfish-routing} proved that in every non-atomic selfish routing game, the cost of an equilibrium is upper bounded by that of an optimal 
solution that routes twice as much traffic. In this section, we recover this result by the mean of linear programming duality. 
Resource augmentation have been widely studied in many contexts in algorithms. Recently, 
\citet{LucarelliThang16:Online-Non-preemptive} have presented an unified 
approach to study resource augmentation in online (scheduling) problems based on primal-dual techniques. We will follow this 
framework to prove the resource augmentation result in non-atomic congestion games. 

Let $(G, (1 + r)w,\ell)$ for some constant $r$ be a non-atomic congestion game 
in which the total amount for players of type $i$ is $(1+r)w_{i}$ and 
the cost function on each resource $e$ is $\ell_{e}$.  
Our purpose is to bound the cost of an arbitrary equilibrium in $(G,w,\ell)$ by that of an optimal solution 
in $(G,(1+r)w,\ell)$ for some $r > 0$. Consider the following formulation (similar to the previous section) $(\mathcal{P}_{r})$
for $(G,(1+r)w,\ell)$. By weak duality, the optimal cost in $(G,(1+r)w,\ell)$ is at least the objective of a dual feasible solution in $(\mathcal{D}_{r})$.

\begin{minipage}[t]{0.4\textwidth}
\begin{align}	
\min ~ \sum_{e} w(T) & \ell_{e}(w(T)) z_{eT} 	\tag{$\mathcal{P}_{r}$}  \\
\sum_{j,k} a_{k} x_{ijk} &\geq (1+r) \cdot w_{i} & &  \forall i	\notag \\
\sum_{T} z_{eT} &= 1 & & \forall e 	\notag \\
\sum_{T: (i,k) \in T} z_{eT}	&= \sum_{j: e \in s_{ij}} x_{ijk} & & \forall (i,k), e \notag \\
x_{ijk}, z_{eT} &\in \{0,1\} & & \forall i,j,e,T \notag \\ \notag
\end{align}
\end{minipage}
\qquad
\begin{minipage}[t]{0.5\textwidth}
\begin{align}
\max \sum_{i} (1+r)w_{i} \alpha_{i} &+ \sum_{e} \beta_{e} 	\tag{$\mathcal{D}_{r}$} \\
a_{k} \alpha_{i} &\leq \sum_{e: e \in s_{ij}} \gamma_{i,k,e}  & &  \forall i,k,j	\notag \\
\beta_{e} + \sum_{(i,k) \in T} \gamma_{i,k,e}	&\leq w(T) \ell_{e}(w(T))  & & \forall e,T \notag \\
\alpha_{i} &\geq 0 & & \forall i \notag \\ \notag
\end{align}
\end{minipage}

Hence, our scheme consists of bounding the cost of an arbitrary equilibrium in $(G,w,\ell)$ and the objective $(\mathcal{D}_{r})$
of an appropriate dual feasible solution. 

\begin{theorem}
In every non-atomic congestion game, for any constant $r > 0$, 
the cost of an equilibrium in $(G,w,\ell)$ is at most $1/r$ times that of an optimal solution in $(G,(1+r)w,\ell)$.
\end{theorem}
\begin{proof}
Given a coarse correlated equilibrium $\vect{\sigma}$ of the game where the amount for players of type $i$ is $w_{i}$.
Construct the dual feasible solution for $(\mathcal{D}_{r})$ as in the proof of Theorem~\ref{thm:non-atomic}. 
As the dual constraints of $(\mathcal{D}_{r})$ and $(\mathcal{D}_{0})$ are the same, the construction in 
the proof of Theorem~\ref{thm:non-atomic} gives a dual feasible solution for $(\mathcal{D}_{r})$.
It remains to bound the objective of 
$(\mathcal{D}_{r})$ of this dual solution to the cost of equilibrium $\vect{\sigma}$, which is 
$\E_{\vect{f} \sim \vect{\sigma}} \left[ \sum_{e} f_{e} \ell_{e}(f_{e})  \right] $. The former is 
\begin{align*}
\sum_{i} (1+r) w_{i} \alpha_{i} + \sum_{e} \beta_{e}
&=  \E_{\vect{f} \sim \vect{\sigma}} \left[ \sum_{e} \biggl( (1+r) \cdot f_{e} \ell_{e}(f_{e})  + v_{e} \ell_{e}(v_{e})
				- v_{e} \ell_{e}(f_{e}) \biggr) \right] 	\\
&\geq  \E_{\vect{f} \sim \vect{\sigma}} \left[ \sum_{e} r \cdot f_{e} \ell_{e}(f_{e})  \right] 
\end{align*}
where the inequality holds since $f_{e} \ell_{e}(f_{e})  + v_{e} \ell_{e}(v_{e}) \geq v_{e} \ell_{e}(f_{e})$. Precisely, 
if $f_{e} \geq v_{e}$ then $f_{e} \ell_{e}(f_{e}) \geq v_{e} \ell_{e}(f_{e})$ and 
if $f_{e} < v_{e}$ then $ v_{e} \ell_{e}(v_{e}) > v_{e} \ell_{e}(f_{e})$ (since $\ell_{e}$ in non-decreasing).
Hence, we deduce that the objective of $(\mathcal{D}_{r})$ is at least $r$ times the cost of equilibrium $\vect{\sigma}$.
\end{proof}

\subsection{Splittable Congestion Games}
\paragraph{Model.} In this section we consider the splittable congestion games also in discrete setting. Fix a constant $\epsilon > 0$ (arbitrarily small).
In a splittable congestion game, there is a set $E$ of resources, each resource is associated to 
a non-decreasing differentiable cost function $\ell_{e}: \mathbb{R}^{+} \rightarrow \mathbb{R}^{+}$ such that 
$x \ell_{e}(x)$ is convex. 
There are $n$ players, a player $i$ has a set of strategies $\mathcal{S}_{i}$ and has weight $w_{i}$, a multiple of $\epsilon$.
A strategy of player $i$ is a distribution $u^{i}$ of its weight $w_{i}$ among strategies $s_{ij}$ in $\mathcal{S}_{i}$
such that $\sum_{s_{ij} \in \mathcal{S}_{i}} u^{i}_{s_{ij}} = w_{i}$ and $u^{i}_{s_{ij}} \geq 0$ is a multiple of $\epsilon$. 
A strategy profile is a vector $\vect{u} = (u^{1}, \ldots, u^{n})$ of all players' strategies. 
We abuse notation and define $u^{i}_{e} = \sum_{e \in s_{ij}} u^{i}_{s_{ij}}$ 
as the load player $i$ distributes on resource $e$ and $u_{e} = \sum_{i=1}^{n} u^{i}_{e}$ the total load on $e$. 
Given a strategy profile $\vect{u}$, the cost of player $i$ is defined as $C_{i}(\vect{u}) := \sum_{e} u^{i}_{e} \cdot \ell_{e}(u_{e})$.
A strategy profile $\vect{u}$ is a pure Nash equilibrium if and only if for 
every player $i$ and all $s_{ij}, s_{ij'} \in \mathcal{S}_{i}$ with $u^{i}_{s_{ij}} > 0$:
$$
\sum_{e \in s_{ij}} \bigl( \ell_{e}(u_{e}) + u^{i}_{e} \cdot \ell'_{e}(u_{e}) \bigr)
\leq \sum_{e \in s_{ij'}} \bigl( \ell_{e}(u_{e}) + u^{i}_{e} \cdot \ell'_{e}(u_{e}) \bigr)
$$
The proof of this equilibrium characterization can be found in \cite{Harks11:Stackelberg-strategies}.
Again, the more general concepts of mixed, correlated and coarse correlated equilibria are defined 
similarly as in Section~\ref{sec:smooth}. 
In the game, the social cost is defined as $C(\vect{u}) := \sum_{i=1}^{n} C_{i}(\vect{u}) = \sum_{e} u_{e} \ell_{e}(u_{e})$.

%\subsubsection{Efficiency of Splittable Congestion Games}

The PoA bounds has been recently established for a large class of cost functions by 
\citet{RoughgardenSchoppmann15:Local-smoothness}. The authors proposed a \emph{local smoothness} framework
and showed that the local smoothness arguments give optimal PoA bounds for a large class of cost functions in 
splittable congestion games. Prior to \citet{RoughgardenSchoppmann15:Local-smoothness}, the works of 
\citet{CominettiCorrea09:The-impact-of-oligopolistic} and \citet{Harks11:Stackelberg-strategies} 
have also the flavour of local smoothness though their bounds are not tight. The local smooth arguments
extends to the correlated equilibria of a game but not to the coarse correlated equilibria. Motivating by the duality approach, 
we define a new notion of smoothness and prove a bound on the PoA of coarse correlated equilibria.
It turns out that this PoA bound for coarse correlated equilibria is indeed \emph{tight} for all classes of scale-invariant cost functions by 
the lower bound given by \citet[Section 5]{RoughgardenSchoppmann15:Local-smoothness}. 
A class of cost function $\mathcal{L}$ is \emph{scale-invariant} if $\ell \in \mathcal{L}$ implies that 
$a \cdot \ell(b \cdot x) \in \mathcal{L}$ for every $a,b > 0$.

\paragraph{Formulation.}
Given a splittable congestion game, we formulate the problem by the same configuration program for non-atomic congestion game.
Denote a finite set of multiples of $\epsilon$ as $\{a_{0}, a_{1}, \ldots, a_{m}\}$ where $a_{k} = k\cdot \epsilon$
and $m = \max_{i=1}^{n} w_{i}/\epsilon$.  
We say that $T$ is a \emph{configuration} of a resource $e$ if $T$ consists of couples $(i,k)$ that specifies 
player $i$ distributes an amount $a_{k}$ of its weight $w_{i}$ to a strategy $s_{ij} \in S_{i}$ where 
$e \in s_{ij}$. 
%Note that in a configuration $T$ of a resource $e$, 
%there might be multiple couples $(i,k) \in T$ and $(i,k') \in T$ corresponding to the same player. 
%It simply means that player $i$ distributes the amounts $a_{k}$ and $a_{k'}$ to some strategies $s_{ij}$ and 
%$s_{ij'}$ respectively that contains resource $e$, i.e., $e \in s_{ij}$ and $e \in s_{ij'}$.
Intuitively, a configuration of a resource is a strategy profile of a game restricted on the resource. 
Let $x_{ijk}$ be variable indicating whether player $i$ distributes an amount $a_{k}$ of  
its weight to strategy $s_{ij} \in S_{i}$.
For every resource $e$ and a configuration $T$, let $z_{eT}$ be a variable such that 
$z_{eT} = 1$ if and only if for $(i,k) \in T$, $x_{ijk} = 1$ 
for some $s_{ij} \in S_{i}$ such that $e \in s_{ij}$.
For a configuration $T$ on resource $e$, denote $w(T)$ the total amount distributed by players in $T$ on $e$.

\begin{minipage}[t]{0.45\textwidth}
\begin{align*}
\min ~ \sum_{e} w(T) & \ell_{e}(w(T)) z_{eT} \\
\sum_{j,k} a_{k} x_{ijk} &\geq w_{i} & &  \forall i\\
\sum_{T} z_{eT} &= 1 & & \forall e \\
\sum_{T: (i,k) \in T} z_{eT}	&= \sum_{j: e \in s_{ij}} x_{ijk} & & \forall (i,k), e \\
x_{ij}, z_{eT} &\in \{0,1\} & & \forall i,j,e,T \\
\end{align*}
\end{minipage}
\qquad
\begin{minipage}[t]{0.5\textwidth}
\begin{align*}
\max \sum_{i} w_{i} \alpha_{i} &+ \sum_{e} \beta_{e} \\
a_{k} \alpha_{i} &\leq \sum_{e: e \in s_{ij}} \gamma_{i,k,e}  & &  \forall i,k,j\\
\beta_{e} + \sum_{(i,k) \in T} \gamma_{i,k,e}	&\leq w(T) \ell_{e}(w(T))  & & \forall e,T \\
\alpha_{i} &\geq 0 & & \forall i \\
\end{align*}
\end{minipage}
Again, in the primal, the first constraint says that a player $i$ distributes the total weight $w_{i}$ among its strategies.
The second constraint means that a resource $e$ is always associated to a configuration (possibly empty). 
The third constraint guarantees that if a player $i$ distributes an amount $a_{k}$ to some strategy $s_{ij}$ containing 
resource $e$ then there must be a configuration $T$ such that $(i,k) \in T$ and $z_{eT} = 1$. 

All previous duality proofs have the same structure: in the dual LP, the first constraint gives the characterization of 
an equilibrium and the second one settles the PoA bounds. Following this line, we give the following definition. 

\begin{definition}
A cost function $\ell: \mathbb{R}^{+} \rightarrow \mathbb{R}^{+}$ is \emph{$(\lambda,\mu)$-dual-smooth} if for every
vectors $\vect{u} = (u_{1}, \ldots, u_{n})$ and $\vect{v} = (v_{1}, \ldots, v_{n})$, 
$$
 v \ell(u) + \sum_{i=1}^{n} u_{i}(v_{i} - u_{i}) \cdot \ell'(u) \leq \lambda \cdot v\ell(v) + \mu \cdot u\ell(u)
$$  
where $u = \sum_{i} u_{i}$ and $v = \sum_{i} v_{i}$.
A splittable congestion game is \emph{$(\lambda,\mu)$-dual-smooth} if every resource $e$ in the game, 
function $\ell_{e}$ is $(\lambda,\mu)$-dual-smooth.
\end{definition}

\begin{theorem}	\label{thm:splittable-locally-smooth}
For every $(\lambda,\mu)$-dual-smooth splittable congestion game $G$, the price of anarchy of coarse correlated equilibria of $G$
is at most $\lambda/(1-\mu)$. This bound is \emph{tight} for the class of scalable cost functions.
\end{theorem}
\begin{proof}
The proof follows the duality scheme. 

\paragraph{Dual Variables.} 
Fix parameter $\lambda$ and $\mu$.
Given a coarse correlated equilibrium $\vect{\sigma}$, define corresponding dual variables as follows.
\begin{align*}
\alpha_{i} &= \frac{1}{\lambda} \E_{\vect{u} \sim \vect{\sigma}} \biggl[ \sum_{e \in s_{ij} } \ell_{e}(u_{e}) + u^{i}_{e} \ell'_{e}(u_{e}) \biggr] 
\textnormal{ for some } s_{ij} \in \mathcal{S}_{i}: u^{i}_{s_{ij}} > 0, \\
\beta_{e} &= - \frac{1}{\lambda} \E_{\vect{u} \sim \vect{\sigma}} \biggl[ \mu \cdot u_{e} \ell_{e}(u_{e}) + \sum_{i} (u^{i}_{e})^{2} \cdot \ell'_{e}(u_{e}) \biggr], \\
\gamma_{i,k,e} &= \frac{1}{\lambda} \E_{\vect{u} \sim \vect{\sigma}} \bigl[ a_{k} \bigl( \ell_{e}(u_{e}) + u^{i}_{e} \ell'_{e}(u_{e}) \bigr) \bigr].
\end{align*}
The dual variables have similar interpretations as previous analysis. 
Up to some constant factors, variable $\alpha_{i}$ is the marginal cost of a strategy used by player $i$ in the equilibrium;
%$\beta_{e}$ stands for the total cost of players on resource $e$ in this equilibrium; 
and 
$\gamma_{i,k,e}$ represents an estimation of the cost of player $i$ on resource $e$ if player $i$ distributes
an amount $a_{k}$ of its weight to some strategy containing $e$ while players $i'$ other than $i$ follows their strategies
in the equilibrium. 

\paragraph{Feasibility.} By this definition of dual variables, the first dual constraint holds 
since it is the definition of coarse correlated equilibrium. 
Rearranging the terms, the second dual constraint for a resource $e$ and a configuration $T$ reads
\begin{align*}
\frac{1}{\lambda} \sum_{(i,k) \in T} \E_{\vect{u} \sim \vect{\sigma}} 
	\bigl[ a_{k} \cdot \ell_{e}(u_{e}) + u^{i}_{e} (a_{k} - u^{i}_{e}) \ell'_{e}(u_{e}) \bigr) \bigr]
	\leq w(T) \ell_{e}(w(T)) + \frac{\mu}{\lambda} \E_{\vect{u} \sim \vect{\sigma}} \bigl[ u_{e} \ell_{e}(u_{e}) \bigr]
\end{align*}
This inequality follows directly from the definition of $(\lambda,\mu)$-dual-smoothness and
linearity of expectation (and note that $w(T) \ell_{e}(w(T)) = \E_{\vect{u} \sim \vect{\sigma}} \bigl[ w(T) \ell_{e}(w(T)) \bigr] $
and $w(T) = \sum_{(i,k) \in T} a_{k}$).

\paragraph{Bounding primal and dual.} 
By the definition of dual variables, the dual objective is
\begin{align*}	
\sum_{i} w_{i} \alpha_{i} &+ \sum_{e} \beta_{e}
= \sum_{e} \biggl( \sum_{i} u^{i}_{e} \alpha_{i} + \beta_{e} \biggr) \\
&= \frac{1}{\lambda} \E_{\vect{u} \sim \vect{\sigma}} \biggl[ \sum_{e} u_{e} \ell_{e}(u_{e}) + \sum_{i} (u^{i}_{e})^{2} \cdot \ell'_{e}(u_{e}) \biggr]
	-  \frac{1}{\lambda} \E_{\vect{u} \sim \vect{\sigma}} \biggl[ \mu \cdot u_{e} \ell_{e}(u_{e}) + \sum_{i} (u^{i}_{e})^{2} \cdot \ell'_{e}(u_{e}) \biggr] \\ 
&= \frac{1 - \mu}{\lambda}  \E_{\vect{u} \sim \vect{\sigma}} \biggl[ \sum_{e} u_{e}\ell_{e}(u_{e}) \biggr]	
\end{align*}
while the cost of the equilibrium $\vect{\sigma}$ is 
$\E_{\vect{u} \sim \vect{\sigma}} \bigl[ \sum_{e} u_{e} \ell_{e}(u_{e}) \bigr]$.
The theorem follows.
\end{proof}

\section{Efficiency in Welfare Maximization}	\label{sec:welfare}
In a general mechanism design setting, each player $i$ has a set of actions $\mathcal{A}_{i}$
for $1 \leq i \leq n$. Given an action $a_{i} \in \mathcal{A}_{i}$ chosen by each player $i$ for $1 \leq i \leq n$, which lead to the 
action profile $\vect{a} = (a_{1}, \ldots, a_{n}) \in \mathcal{A} = \mathcal{A}_{1} \times \ldots \times \mathcal{A}_{n}$, 
the auctioneer decides an outcome $o(\vect{a})$ among the set of feasible outcomes $\mathcal{O}$. 
Each player $i$ has a \emph{valuation} (or \emph{type}) $v_{i}$ taking values in a parameter space $\mathcal{V}_{i}$.
For each outcome $o \in \mathcal{O}$, player $i$ has \emph{utility} $u_{i}(o,v_{i})$ depending on
the outcome of the game and its valuation $v_{i}$. 
Since the outcome $o(\vect{a})$ of the game is determined by the action profile $\vect{a}$, 
the utility of a player $i$ is denoted as $u_{i}(\vect{a};v_{i})$. 
We are interested in auctions that in general consist of an allocation rule and a payment rule.
Given an action profile $\vect{a} = (a_{1}, \ldots, a_{n})$, the auctioneer decides an allocation 
and a payment $p_{i}(\vect{a})$ for each player $i$. 
Then, the \emph{utility} of player $i$ with valuation $v_{i}$, following the quasi-linear utility model, 
is defined as $u_{i}(\vect{a};v_{i}) = v_{i} - p_{i}(\vect{a})$. 
The \emph{social welfare} of an auction is defined as the total utility of all participants (the players and the auctioneer):  
$\textsc{Sw}(\vect{a};\vect{v}) =  \sum_{i=1}^{n} u_{i}(\vect{a};v_{i}) + \sum_{i=1}^{n} p_{i}(\vect{a})$. 

In the paper, we consider incomplete-information settings. 
In the settings, the valuation $v_{i}$ of each player is a private information and is drawn independently from a publicly known
distribution $\vect{F}$ with density function $\vect{f}$. 
Let $\Delta(\mathcal{A}_{i})$ be the set of probability distributions over the actions 
in $\mathcal{A}_{i}$.  A strategy of a player is a mapping $\sigma_{i}: \mathcal{V}_{i} \rightarrow \Delta(\mathcal{A}_{i})$
from a valuation $v_{i} \in \mathcal{V}_{i}$ to a distribution over actions $\sigma_{i}(v_{i}) \in \Delta(\mathcal{A}_{i})$.

\begin{definition}[Bayes-Nash equilibrium]
A strategy profile $\vect{\sigma} = (\sigma_{1}, \ldots, \sigma_{n})$ is a 
\emph{Bayes-Nash equilibrium (BNE)}  if for every player $i$, for every valuation $v_{i} \in \mathcal{V}_{i}$, and 
for every action $a'_{i} \in \mathcal{A}_{i}$:
\begin{align*}
\E_{\vect{v}_{-i} \sim \vect{F}_{-i}(v_{i})} \left[ \E_{\vect{a} \sim \vect{\sigma}(\vect{v})} \left[ u_{i}(\vect{a};v_{i}) \right] \right]
\geq 
\E_{\vect{v}_{-i} \sim \vect{F}_{-i}(v_{i})} \left[ \E_{\vect{a}_{-i} \sim \vect{\sigma}_{-i}(\vect{v}_{-i})} \left[ u_{i}(a'_{i},\vect{a}_{-i};v_{i}) \right] \right]
\end{align*}
\end{definition}
For a vector $\vect{w}$, we use $\vect{w}_{-i}$ to denote the vector $\vect{w}$ with the $i$-th component removed. 
Besides, $\vect{F}_{-i}(v_{i})$ stands for the probability distribution over all players other than $i$
conditioned on the valuation $v_{i}$ of player $i$. 

The price of anarchy of Bayes-Nash equilibria of an auction is defined as 
$$
\inf_{\vect{F}, \vect{\sigma}} \frac{\E_{\vect{v} \sim \vect{F}}\bigl[ \E_{\vect{a} \sim \vect{\sigma}(\vect{v})}[\textsc{Sw}(\vect{a};\vect{v})] \bigr]}{\E_{\vect{v} \sim \vect{F}} \bigl[ \textsc{Opt}(\vect{v})\bigr]}
$$
where the infimum is taken over Bayes-Nash equilibria $\vect{\sigma}$ and 
$\textsc{Opt}(\vect{v})$ is the optimal welfare with valuation profile $\vect{v}$.

In the paper, we consider discrete settings of valuations and payments, i.e., 
there are only a finite (large) number of possible valuations
and payments. The main purpose of restricting to discrete settings is that we can use tools from linear programming.
The continuous settings can be done by considering successively finer discrete spaces.

\subsection{Smooth Auctions}	\label{sec:smooth-auction}

In this section, we show that the primal-dual approach also captures the smoothness framework in 
studying the inefficiency of Bayes-Nash equilibria in incomplete-information settings. 
Smooth auctions have been defined by \citet{Roughgarden15:The-price-of-anarchy} 
and \citet{SyrgkanisTardos13:Composable-and-efficient}. The definitions are slightly different but both are inspired 
by the original smoothness argument \cite{Roughgarden15:Intrinsic-robustness} and 
all known smoothness-based proofs can be equivalently analyzed by one of these definitions.  
In this section, we consider the definition of smooth auctions in \cite{Roughgarden15:The-price-of-anarchy}
and revisit the price of anarchy bound of smooth auctions.
In the end of the section, we show that a similar proof carries through the smooth auctions defined by
\citet{SyrgkanisTardos13:Composable-and-efficient}.

\begin{definition}[\cite{Roughgarden15:The-price-of-anarchy}]	\label{def:smooth-auctions}
For parameters $\lambda, \mu \geq 0$, an auction is $(\lambda,\mu)$-\emph{smooth} if for every valuation profile 
$\vect{v} = (v_{1}, \ldots, v_{n})$,
there exists action distribution $D^{*}_{1}(\vect{v}), \ldots,D^{*}_{n}(\vect{v})$ over $\mathcal{A}_{1}, \ldots, \mathcal{A}_{n}$
such that, for every action profile $\vect{a}$,
\begin{align}	\label{eq:smooth-auctions}
\sum_{i} \E_{a^{*}_{i} \sim D^{*}_{i}(\vect{v})} \bigl[ u_{i}(a^{*}_{i},\vect{a}_{-i}; v_{i}) \bigr]
	\geq \lambda \cdot \textsc{Sw}(\vect{a}^{*};\vect{v}) - \mu \cdot \textsc{Sw}(\vect{a};\vect{v})  
\end{align}
\end{definition}

\begin{theorem}[\cite{Roughgarden15:The-price-of-anarchy}]	\label{thm:smooth-auction-Roughgarden}
If an auction is $(\lambda,\mu)$-smooth and the distributions of player valuations are independent then
every Bayes-Nash equilibrium has expected welfare at least $\frac{\lambda}{1+\mu}$ times the optimal 
expected welfare. 
\end{theorem}
\begin{proof}
Given an auction, we formulate the corresponding optimization problem by a configuration LP.  
A \emph{configuration} $A$ consists of pairs $(i,a_{i})$ such that $(i,a_{i}) \in A$ means that in configuration 
$A$, player $i$ chooses action $a_{i}$. Intuitively, a configuration is an action profile of players. 
For every player $i$, every valuation $v_{i} \in \mathcal{V}_{i}$ 
and every action $a_{i} \in \mathcal{A}_{i}$, let $x_{i,a_{i}}(v_{i})$ be the variable representing the probability that player $i$ chooses
action $a_{i}$. Besides, for every valuation profile $\vect{v}$, let $z_{A}(\vect{v})$ be the variable indicating the probability that 
the chosen configuration (action profile) is $A$. 
For each configuration $A$ and valuation profile $\vect{v}$, the auctioneer outcomes an allocation and a payment and 
that results in a social welfare denoted as $c_{A}(\vect{v})$. In the other words, if $\vect{a}$ is the action profile corresponding 
to the configuration $A$ then $c_{A}(\vect{v})$ is in fact $\textsc{Sw}(\vect{a};\vect{v})$. Consider the following formulation and its dual.

\begin{minipage}[t]{0.4\textwidth}
\begin{align*}
\max ~ \sum_{v} c_{A}(\vect{v}) &z_{A}(\vect{v}) \\
\sum_{a_{i} \in \mathcal{A}_{i} } x_{i,a_{i}}(v_{i}) &\leq f_{i}(v_{i}) \quad \forall i,v_{i}\\
\sum_{A} z_{A}(\vect{v}) &\leq f(\vect{v}) \quad  \forall \vect{v} \\
\sum_{A: (i,a_{i}) \in A} z_{A}(v_{i},\vect{v}_{-i})	&\leq f_{-i}(\vect{v}_{-i}) \cdot x_{i,a_{i}}(v_{i})  \\ 
		 & \qquad \quad \forall i,a_{i}, v_{i}, \vect{v}_{-i}  \\
x_{i,a_{i}}(v_{i}), z_{A}(\vect{v}) &\geq 0 \quad \forall i,a_{i},A,v_{i}, \vect{v} \\
\end{align*}
\end{minipage}
\qquad
\begin{minipage}[t]{0.5\textwidth}
\begin{align*}
\min \sum_{i,v_{i}} f_{i}(v_{i}) \cdot \alpha_{i}(v_{i}) +& \sum_{\vect{v}} f(\vect{v}) \cdot \beta(\vect{v}) \\
\alpha_{i}(v_{i}) \geq \sum_{\vect{v}_{-i}} f_{-i}(\vect{v}_{-i}) & \cdot \gamma_{i,a_{i}}(v_{i},\vect{v}_{-i})  & &  \forall i,a_{i},v_{i}\\
\beta(\vect{v}) + \sum_{(i,a_{i}) \in A} \gamma_{i,a_{i}}(\vect{v})	&\geq c_{A}(\vect{v}) & & \forall A, \vect{v}\\
\alpha_{i}(v_{i}), \beta(\vect{v}), \gamma_{i,a_{i}}(\vect{v}) &\geq 0 & & \forall i,v_{i},\vect{v} \\
\end{align*}
\end{minipage}
In the primal, the first and second constraints guarantee that variables $x$ and $z$ represent indeed
the probability distribution of each player and the joint distribution, respectively. 
The third constraint makes the connection between variables $x$ and $z$. 
It ensures that if a player $i$ with valuation $v_{i}$ selects some action $a_{i}$ 
then in the valuation profile $(v_{i},\vect{v}_{-i})$, the probability that the configuration $A$ contains $(i,a_{i})$
must be $f_{-i}(\vect{v}_{-i}) \cdot x_{i,a_{i}}(v_{i})$. The primal objective is the expected welfare of the auction.

\paragraph{Construction of dual variables.}
Assuming that the auction is $(\lambda,\mu)$-smooth. Fix the parameters $\lambda$ and $\mu$.
Given an arbitrary Bayes-Nash equilibrium $\vect{\sigma}$, define dual variables as follows. 
%Note that if the 
%valuation profile $v$ is clear from the context, we simple write $\vect{\sigma}$ instead of $\vect{\sigma}(\vect{v})$.
%
\begin{align*}
\alpha_{i}(v_{i}) &:= \frac{1}{\lambda} \E_{\vect{v}_{-i}} \bigl[ \E_{\vect{b} \sim \vect{\sigma}(v_{i},\vect{v}_{-i})}[u_{i}(\vect{b}; v_{i})] \bigr], \\
\beta(\vect{v}) &:= \frac{\mu}{\lambda} \E_{\vect{b} \sim \vect{\sigma}(\vect{v})} \bigl[ \textsc{Sw}(\vect{b};\vect{v}) \bigr], \\
\gamma_{i,a_{i}}(\vect{v}) &:= \frac{1}{\lambda} \E_{\vect{b}_{-i} \sim \vect{\sigma}_{-i}(\vect{v}_{-i})} [u_{i}(a_{i},\vect{b}_{-i}; v_{i})].
\end{align*}
Informally, up to some constant factors depending on $\lambda$ and $\mu$, 
$\alpha_{i}(v_{i})$ is the expected utility of player $i$ in equilibrium $\vect{\sigma}$;
$\beta(\vect{v})$ stands for the social welfare of the auction where 
the valuation profile is $\vect{v}$ and players follow the equilibrium actions $\vect{\sigma}(\vect{v})$; 
and $\gamma_{i,a_{i}}(\vect{v})$ represents the utility of player $i$ in valuation profile $\vect{v}$ if player $i$ chooses
action $a_{i}$ while other players $i' \neq i$ follows their equilibrium strategies $\vect{\sigma}_{-i}(\vect{v}_{-i})$.

\paragraph{Feasibility.} We show that the constructed dual variables form a feasible solution. 
By the definition of dual variables, the first dual constraint reads
\begin{align*}
\frac{1}{\lambda} \E_{\vect{v}_{-i}} \bigl[ \E_{\vect{b} \sim \vect{\sigma}(\vect{v})}[u_{i}(\vect{b}; v_{i})] \bigr]
&\geq \frac{1}{\lambda} \sum_{\vect{v}_{-i}} f_{-i}(\vect{v}_{-i}) \cdot 
			\E_{\vect{b}_{-i} \sim \vect{\sigma}_{-i}(\vect{v}_{-i})} [u_{i}(a_{i},\vect{b}_{-i}; v_{i})] \\
&= \frac{1}{\lambda} \E_{\vect{v}_{-i}} \bigl[
			\E_{\vect{b}_{-i} \sim \vect{\sigma}_{-i}(\vect{v}_{-i})} [u_{i}(a_{i},\vect{b}_{-i}; v_{i})] \bigr]
\end{align*}
This is exactly the definition that $\vect{\sigma}$ is a Bayes-Nash equilibrium.

For every valuation profile $\vect{v} = (v_{1}, \ldots, v_{n})$ and for any configuration 
$A$ (corresponding action profile $\vect{a} = (a_{1}, \ldots, a_{n})$), the second constraint reads:
\begin{align}	\label{eq:auction-smooth-feasibility}
\frac{\mu}{\lambda} \E_{\vect{b} \sim \vect{\sigma}(\vect{v})} \bigl[ \textsc{Sw}(\vect{b};\vect{v}) \bigr]
+ \sum_{(i,a_{i}) \in A}\frac{1}{\lambda} \E_{\vect{b}_{-i} \sim \vect{\sigma}_{-i}(\vect{v}_{-i})} [u_{i}(a_{i},\vect{b}_{-i}; v_{i})] 
\geq \textsc{Sw}(\vect{a};\vect{v}).
\end{align}
Note that we can write $\textsc{Sw}(\vect{a};\vect{v}) = \E_{\vect{b} \sim \vect{\sigma}(\vect{v})} \bigl[\textsc{Sw}(\vect{a};\vect{v}) \bigr]$.
For any fixed realization $\vect{b}$ of $\vect{\sigma}(\vect{v})$, by $(\lambda,\mu)$-smoothness
\begin{align*}
\frac{\mu}{\lambda} \textsc{Sw}(\vect{b};\vect{v}) 
+ \sum_{i}\frac{1}{\lambda} u_{i}(a_{i},\vect{b}_{-i}; v_{i}) 
\geq \textsc{Sw}(\vect{a};\vect{v}).
\end{align*}
Hence, by taking expectation over $\vect{\sigma}(\vect{v})$, Inequality (\ref{eq:auction-smooth-feasibility}) follows.

\paragraph{Price of Anarchy.} 
The welfare of equilibrium $\vect{\sigma}$ is 
$\E_{\vect{v}} \E_{\vect{b} \sim \vect{\sigma}(\vect{v})} \bigl[ \textsc{Sw}(\vect{b};\vect{v}) \bigr]$ 
while the dual objective of the constructed dual variables is 
\begin{align*}
\sum_{i,v_{i}} f_{i}(v_{i}) &\cdot \frac{1}{\lambda} \E_{\vect{v}_{-i}} \bigl[ \E_{\vect{b} \sim \vect{\sigma}(\vect{v})}[u_{i}(\vect{b}; v_{i})] \bigr] 
+ \sum_{\vect{v}} f(\vect{v}) \cdot \frac{\mu}{\lambda} \E_{\vect{b} \sim \vect{\sigma}(\vect{v})} \bigl[ \textsc{Sw}(\vect{b};\vect{v}) \bigr] \\
&\leq \frac{1+\mu}{\lambda} \cdot \E_{\vect{v}} \E_{\vect{b} \sim \vect{\sigma}(\vect{v})} \bigl[ \textsc{Sw}(\vect{b};\vect{v}) \bigr]
\end{align*}
Therefore, the PoA of a $(\lambda,\mu)$-smooth auction is at most $\lambda/(1+\mu)$.
\end{proof}

\paragraph{Remark.} The notion of $(\lambda,\mu)$-smooth auctions due to \citet{SyrgkanisTardos13:Composable-and-efficient} 
is defined similarly as Definition~\ref{def:smooth-auctions}
but now the parameter $\mu \geq 1$ and Inequality~(\ref{eq:smooth-auctions}) is replaced by the following inequality:
\begin{align}	\label{eq:smooth-auctions-tardos}
\sum_{i} \E_{a^{*}_{i} \sim D^{*}_{i}(\vect{v})} \bigl[ u_{i}(a^{*}_{i},\vect{a}_{-i}; v_{i}) \bigr]
	\geq \lambda \cdot \textsc{Opt}(\vect{v}) - \mu \cdot \textsc{R}(\vect{a})  
\end{align}
where $\textsc{R}(\vect{a})$ is the total payment of players if the action profile is $\vect{a}$. 
Note that, in order to bound the price of anarchy, Inequality~(\ref{eq:smooth-auctions-tardos}) can be replaced by a weaker one, which is:
\begin{align}	\label{eq:smooth-auctions-weak}
\sum_{i} \E_{a^{*}_{i} \sim D^{*}_{i}(\vect{v})} \bigl[ u_{i}(a^{*}_{i},\vect{a}_{-i}; v_{i}) \bigr]
	\geq \lambda \cdot \textsc{Sw}(\vect{a}^{*};\vect{v}) - \mu \cdot \textsc{R}(\vect{a})  
\end{align}

Using the same proof structure of Theorem~\ref{thm:smooth-auction-Roughgarden}, we can prove that the price of anarchy is at most $\lambda/\mu$ \cite{SyrgkanisTardos13:Composable-and-efficient}.
Specifically, define dual variables $\alpha$ and $\gamma$ as previous and 
$$
\beta(\vect{v}) = \frac{\mu}{\lambda} \E_{\vect{b} \sim \vect{\sigma}(\vect{v})} \bigl[ \textsc{R}(\vect{b}) \bigr]
$$ 
The feasibility follows the definitions of Bayes-Nash equilibria and smooth auctions, in particular Inequality (\ref{eq:smooth-auctions-weak}). 
To bound the price of anarchy, 
as $\mu \geq 1$, we have
\begin{align*}
\sum_{i,v_{i}} f_{i}(v_{i}) &\cdot \frac{1}{\lambda} \E_{\vect{v}_{-i}} \bigl[ \E_{\vect{b} \sim \vect{\sigma}(\vect{v})}[u_{i}(\vect{b}; v_{i})] \bigr] 
+ \sum_{\vect{v}} f(\vect{v}) \cdot \frac{\mu}{\lambda} \E_{\vect{b} \sim \vect{\sigma}(\vect{v})} \bigl[ \textsc{R}(\vect{b}) \bigr] 
\leq \frac{\mu}{\lambda} \cdot \E_{\vect{v}} \E_{\vect{b} \sim \vect{\sigma}(\vect{v})} \bigl[ \textsc{Sw}(\vect{b};\vect{v}) \bigr]
\end{align*}
Therefore, the price of anarchy is at most $\lambda/\mu$. 

\subsection{Simultaneous Item-Bidding Auctions} 	\label{sec:welfare-simul}
\paragraph{Model.} In this section, we consider the following Bayesian combinatorial auctions. In the setting, 
there are $m$ items to be sold to $n$ players. Each player $i$ has a private monotone valuation 
$v_{i}: 2^{[m]} \rightarrow \mathbb{R}^{+}$ over different subsets of items $S \subset 2^{[m]}$. 
For simplicity, we denote $v_{i}(S)$ as $v_{iS}$.
The valuation profile $\vect{v} = (v_{1}, \ldots, v_{n})$ is drawn from a \emph{product} distribution $\vect{F}$. 
In other words, the probability distributions $F_{i}$ of valuations $v_{i}$ are independent.  
Designing efficient combinatorial auctions are in general complex and a major direction in literature is to seek 
simple and efficient auctions in term of PoA. Among others, simultaneous item-bidding auctions are of particular interest.  

We consider two forms of simultaneous item-bidding auctions: \emph{simultaneous first-price auctions (S1A)} and 
\emph{simultaneous second-price auctions (S2A)}. In the auctions, each player submits simultaneously a vector of bids,
one for each item. A typical assumption is \emph{non-overbidding} property in which 
each player submits a vector $b_{i}$ of bids such that for any set of items $S$, $\sum_{j \in S} b_{ij} \leq v_{iS}$.  
Given the bid profile, each item is allocated to the player with highest bid. 
In a simultaneous first-price auction, the payment of the winner of each item is its bid on the item; 
while in a simultaneous second-price auction, the winner of each item pays the second highest bid on the item.

\subsubsection{Connection between Primal-Dual and Non-Smooth Techniques}	\label{sec:pd-non-smooth}
In this section, we consider the setting in which all player valuations are 
\emph{sub-additive}. That is, $v_{i}(S \cup T) \leq v_{i}(S) + v_{i}(T)$ for every player $i$ and
every subsets $S, T \subset 2^{[m]}$. 
The PoA of simultaneous item-bidding auctions has been widely studied in this setting. 
Using smoothness framework in auctions, logarithmic bounds on PoA for S1A and S2A 
are given by \citet{HassidimKaplan11:Non-price-Equilibria} 
and \citet{BhawalkarRoughgarden11:Welfare-guarantees}, respectively. 
Recently, \citet{FeldmanFu13:Simultaneous-auctions} presented a significant improvement by establishing
the PoA bounds 2 and 4 for S1A and S2A, respectively. Their proof arguments go beyond the smoothness framework. 
In the following, we revisit the results of \citet{FeldmanFu13:Simultaneous-auctions} and 
show that the duality approach captures the non-smooth technique in 
\cite{FeldmanFu13:Simultaneous-auctions}.

\paragraph{Formulation.} Given a valuation profile $\vect{v}$, let $\overline{x}_{ij}(\vect{v})$ be the variable indicating 
whether player $i$ receives item $j$ in valuation profile $\vect{v}$. Let $\overline{z}_{iS}(\vect{v})$ be the variable indicating 
whether player $i$ receives a set of items $S$. Then for any profile $\vect{v}$ and for any item $j$, 
$\sum_{i} \overline{x}_{ij}(\vect{v}) \leq 1$, meaning that an item $j$ is allocated to at most one player.
Moreover, $\sum_{S: j \in S} \overline{z}_{iS}(\vect{v}) = \overline{x}_{ij}(\vect{v})$, meaning that 
if player $i$ receives item $j$ then some subset of items $S$ allocated to $i$ must contain $j$. 
Besides, $\sum_{S} \overline{z}_{iS}(\vect{v}) = 1$ since some subset of items (possibly empty) is allocated to $i$.

Let $x_{ij}(v_{i})$ and $z_{iS}(v_{i})$ be \emph{interim} variables corresponding to $\overline{x}_{ij}(\vect{v})$ 
and $\overline{z}_{iS}(\vect{v})$ and are defined as follows:
$$
x_{ij}(v_{i}) := \E_{\vect{v}_{-i} \sim \vect{F}_{-i}} \bigl[ \overline{x}_{ij}(v_{i}, \vect{v}_{-i}) \bigr], 
\qquad \qquad
z_{iS}(v_{i}) :=  \E_{\vect{v}_{-i} \sim \vect{F}_{-i}} \bigl[ \overline{z}_{iS}(v_{i}, \vect{v}_{-i}) \bigr]
$$
where $\vect{F}_{-i}$ is the product distribution of all players other than $i$.
Consider the following relaxation with interim variables and its dual. 
The constraints in the primal follow the relationship between the interim variables $x_{ij}(v_{i}), z_{iS}(v_{i})$
and variables $\overline{x}_{ij}(\vect{v}), \overline{z}_{iS}(\vect{v})$.

\begin{minipage}[t]{0.4\textwidth}
\begin{align*}
\max ~ \sum_{i,S} \sum_{v_{i}} f_{i}(v_{i}) & \bigl[ v_{iS} \cdot  z_{iS}(v_{i}) \bigr] & & \\
\sum_{i} \sum_{v_{i} \in V_{i}} f_{i}(v_{i}) x_{ij}(v_{i}) &\leq 1 & &  \forall j \\
\sum_{S} z_{iS}(v_{i}) &= 1 & & \forall i,v_{i} \\
\sum_{S: j \in S} z_{iS}(v_{i}) &= x_{ij}(v_{i}) & &  \forall i,j, v_{i} \\
x_{ij}(v_{i}), z_{iS}(v_{i}) &\geq 0 & & \forall i,j,S,v_{i}
\end{align*}
\end{minipage}
\qquad
\begin{minipage}[t]{0.5\textwidth}
\begin{align*}
\min ~ \sum_{i,v_{i}} \alpha_{i}(v_{i}) &+ \sum_{j} \beta_{j} \\
f_{i}(v_{i}) \cdot \beta_{j} &\geq \gamma_{i,j}(v_{i}) & & \forall i,j,v_{i} \\
\alpha_{i}(v_{i})  + \sum_{j \in S} \gamma_{i,j}(v_{i}) &\geq f_{i}(v_{i}) \cdot v_{iS} 
						& & \forall i,S,v_{i} \\
\alpha_{i}(v_{i}) & \geq 0 & & \forall i,v_{i} \\
\end{align*}
\end{minipage}

\paragraph{Dual Variables.} Fix a Bayes-Nash equilibrium $\vect{\sigma}$. 
%To simplify notation, assume that $\vect{\sigma}$ is a pure Bayes-Nash equilibrium.
%The same analysis holds for general ones by adding an additional expectation over players'
%random actions. 
Given a valuation $\vect{v}$, denote $\vect{b} = (b_{1}, \ldots, b_{n}) = \vect{\sigma}(\vect{v})$
as the bid equilibrium. Let $\vect{B}$ be the distribution of $\vect{b}$ over the randomness of $\vect{v}$
and $\vect{\sigma}$. Let $\vect{B}(v_{i})$ be the distribution of $\vect{b}$ over the randomness of $\vect{v}$
and $\vect{\sigma}$ while the valuation $v_{i}$ of player $i$ is fixed.
Since $v_{i}$ and $\vect{v}_{-i}$ are independent and each $\sigma_{i}$  is a mapping 
$\mathcal{V}_{i} \rightarrow \Delta(\mathcal{A}_{i})$, strategy $b_{i}$ is independent of $\vect{b}_{-i}$.
Let $\vect{B}_{-i}$ be the distribution of $\vect{b}_{-i}$.
We define dual variables as follows.

Let $\alpha_{i}(v_{i})$ be proportional to the expected utility of player $i$ with valuation $v_{i}$,
over the randomness of valuations $\vect{v}_{-i}$ of other players. Specifically, 
$$
\alpha_{i}(v_{i}) := 2 f_{i}(v_{i}) \cdot \E_{\vect{v}_{-i} \sim \vect{F}_{i}} \bigl[ \E_{\vect{\sigma}} \bigl[ u_{i}\bigl(\vect{\sigma}(v_{i},\vect{v}_{-i}), v_{i} \bigr) \bigr] \bigr]
	= 2 f_{i}(v_{i}) \cdot \E_{\vect{b} \sim \vect{B}(v_{i})} \left[ u_{i}\bigl(\vect{b}, v_{i} \bigr) \right] 
$$ 
Besides, let $\gamma_{i,j}(v_{i})$ be proportional to the expected value of the bid on item $j$ 
if player $i$ with valuation $v_{i}$ want to win item $j$ while other players follow the equilibrium strategies. Formally, 
$$
\gamma_{i,j}(v_{i}) := 2f_{i}(v_{i}) \cdot \E_{\vect{b}_{-i} \sim \vect{B}_{-i}} \left[ \max_{k \neq i} b_{kj}\right]
$$
Finally, define $\beta_{j} :=  2 \max_{i} \E_{\vect{b}_{-i} \sim \vect{B}_{-i}} \left[ \max_{k \neq i} b_{kj}\right]$.
%$$
%\beta_{j} = 2 \max_{i, v_{i}}  \E_{\vect{v}_{-i} \sim \vect{F}_{-i}} \E_{\vect{\sigma}} \left[ \max_{k \neq i} b_{kj}\right]
%= 2 \max_{i} \E_{\vect{v}_{-i} \sim \vect{F}_{-i}} \E_{\vect{\sigma}}  \left[ \max_{k \neq i} b_{kj}\right]
%$$
%where $\vect{b} = \vect{\sigma}(v_{i},\vect{v}_{-i})$.  
%Again, the last equality holds since $\vect{F}$ is a product distribution.

The following lemma shows the feasibility of the variables. The main core of the proof relies
on an argument in \cite{FeldmanFu13:Simultaneous-auctions}.
 
\begin{lemma}
The dual vector $(\alpha,\beta,\gamma)$ defined above constitutes a dual feasible solution. 
\end{lemma}
\begin{proof}
The first dual constraint follows immediately by the definitions of dual variables $\beta$ and $\gamma$. 
We are now proving the second dual constraint. 
Fix a player $i$ with sub-additive valuation $v_{i}$ and assume that $f_{i}(v_{i}) > 0$ (otherwise, it is trivial).
By \cite{FeldmanFu13:Simultaneous-auctions} 
(or see \cite[Lemma 1.3]{Roughgarden:Frontiers-in-Mechanism} for another clear exposition), 
for any set of items $S$, there exists an action $b^{*}_{i}$ 
such that
\begin{align*}
\E_{\vect{b}_{-i} \sim \vect{B}_{-i}} \biggl[u_{i}\bigl( (b^{*}_{i},\vect{b}_{-i}), v_{i}\bigr)\biggr]
+ \E_{\vect{b}_{-i} \sim \vect{B}_{-i}} \biggl[ \sum_{j \in S} \max_{k\neq i} b_{kj} \biggr] 
\geq \frac{1}{2} v_{iS}. 
\end{align*}
Moreover, the first term in the left-hand side is at most the utility of player $i$ with valuation $v_{i}$ 
since $(b_{i},\vect{b}_{-i})$ is a Bayes-Nash equilibrium. Therefore, 
\begin{align*}
\E_{\vect{b} \sim \vect{B}(v_{i})} \left[ u_{i}\bigl(\vect{b},v_{i}\bigr) \right]
+ \E_{\vect{b} \sim \vect{B}(v_{i})} \biggl[ \sum_{j \in S} \max_{k\neq i} b_{kj} \biggr] 
\geq \frac{1}{2} v_{iS} 
\end{align*}
By the definition of dual variables, 
this inequality is exactly the second constraint by multiplying both sides by $2f_{i}(v_{i})$.
\end{proof}

\begin{theorem}[\cite{FeldmanFu13:Simultaneous-auctions}]
If player valuations are sub-additive then every Bayes-Nash equilibrium of a S1A (or S2A) 
has expected welfare at least 1/2 (or 1/4, resp) of the optimal one.
\end{theorem}
\begin{proof}
For an item $j$, let $i^{*}(j) \in \arg \max_{i} \E_{\vect{v}_{-i} \sim \vect{F}_{-i}} \left[ \max_{k \neq i} b_{kj}\right]$.
Hence,
\begin{align*}
\beta_{j} &= 2 \E_{\vect{v}_{-i^{*}(j)} \sim \vect{F}_{-i^{*}(j)}} \E_{\vect{\sigma}} \biggl[ \max_{k \neq i^{*}(j)} b_{kj}\biggr]
= 2 \E_{v_{i^{*}(j)} \sim F_{i}} \E_{\vect{v}_{-i^{*}(j)} \sim \vect{F}_{-i^{*}(j)}} \E_{\vect{\sigma}} \biggl[ \max_{k \neq i^{*}(j)} b_{kj}\biggr] \\
&= 2 \E_{\vect{v} \sim \vect{F}} \E_{\vect{\sigma}} \biggl[ \max_{k \neq i^{*}(j)} b_{kj}\biggr]
\end{align*}
where the second equality is due to the fact that the term 
$E_{\vect{v}_{-i^{*}(j)} \sim \vect{F}_{-i^{*}(j)}} \E_{\vect{\sigma}} \bigl[ \max_{k \neq i^{*}(j)} b_{kj}\bigr]$ is independent of $v_{i^{*}(j)}$.
Therefore, the dual objective is 
\begin{align*}	%\label{eq:dual-simultaneous}
\sum_{i,v_{i}} \alpha_{i}(v_{i}) &+ \sum_{j} \beta_{j}
= 2 \E_{\vect{v} \sim \vect{F}} E_{\vect{\sigma}} \biggl[  \sum_{i} u_{i}(\vect{b},v_{i}) + \sum_{j}\max_{k \neq i^{*}(j)} b_{kj} \biggr]
\end{align*}
Fix a random choice of profile $\vect{v}$ and $\vect{\sigma}$ (so the bid profile $\vect{b}$ is fixed). 
We bound the dual objective, i.e., the right-hand side of the above equality, 
in S1A and S2A. Note that the utility of a player winning no item is 0.

\paragraph{First Price Auction.} 
Partition the set of items into the winning items of each player. 
Consider a player $i$ with the set of winning items $S$. 
The utility of this player $i$ is $v_{iS} - \sum_{j \in S} \max_{k} b_{kj}$. Hence,
$v_{iS} - \sum_{j \in S} b_{ij} + \sum_{j \in S} \max_{k \neq i^{*}(j)} b_{kj} \leq v_{iS}$ since by the allocation rule, 
$b_{ij} = \max_{k} b_{kj}$ for every $j \in S$. Hence, summing over all players, 
the dual objective is bounded by twice the total expected valuation of winning players, which is the primal. 
So the price of anarchy is at most 2.

\paragraph{Second Price Auction.}
Similarly, consider a player $i$ with the set of winning items $S$.  
The utility of player $i$ as well as its payment (by no-overbidding) are at most $v_{iS}$. 
Therefore, summing over all players, the dual objective is bounded by 
four times the total expected valuation of winning players. Hence, the price of 
anarchy is at most 4.
\end{proof}

\paragraph{Remark.} The non-overbidding assumption, a risk-aversion assumption, 
is given in order to prevent players from suffering negative utility while receiving items.
We use this assumption in the proof only in settling the ratio between the primal and the dual;
specifically to argue that the payment of a player does not exceed its valuation on the received items.  
The above analysis holds even without this assumption in the following sense. Assume that 
players are allowed to bid up to a constant $r$ times their valuation 
(hence, players risk to have negative utility). Then, the PoA for S2A 
is $2(1+r)$.  

\subsubsection{Connection between Primal-Dual and No-Envy Learning}
Very recently, \citet{DaskalakisSyrgkanis16:Learning-in-auctions:} have introduced \emph{no-envy learning} --- a novel concept of learning in auctions.
The notion is inspired by the concept of Walrasian equilibrium and it is 
motivated by the fact that no-regret learning algorithms (which converge to coarse correlated equilibria) 
for the simultaneous item-bidding auctions are computationally inefficient as the number of player actions 
are exponential. When the players have fractionally sub-additive (XOS) valuation, 
\citet{DaskalakisSyrgkanis16:Learning-in-auctions:} showed that no-envy outcomes are a relaxation of no-regret outcomes. Moreover, 
no-envy outcomes maintain the approximate welfare optimality of no-regret outcomes while ensuring the 
computational tractability. In this section, we explore the connection between the no-envy learning 
and the primal-dual approach. Indeed, the notion of no-envy learning would be naturally derived from the dual constraints
very much in the same way as the smoothness argument is.

We recall the notion of no-envy learning algorithms \cite{DaskalakisSyrgkanis16:Learning-in-auctions:}. We first define the \emph{online learning problem}.
In the online learning problem, at each step $t$, the player chooses a bid vector $b^{t} = (b^{t}_{1}, \ldots, b^{t}_{m})$
where $b^{t}_{j}$ is the bid on item $j$ for $1 \leq j \leq m$; and the adversary picks
adaptively (depending on the history of the play but not on the current bid $b^{t}$) a threshold vector 
$\theta^{t} = (\theta^{t}_{1}, \ldots, \theta^{t}_{m})$.
The player wins the set $S^{*}(b^{t},\theta^{t}) = \{j: b^{t}_{j} \geq \theta^{t}_{j}\}$ and gets reward:
\begin{align*}
u(b^{t},\theta^{t}) := v\bigl( S^{*}(b^{t},\theta^{t})\bigr) - \sum_{j \in S^{*}(b^{t},\theta^{t})} \theta^{t}_{j}
\end{align*}
where $v: 2^{[m]} \rightarrow \mathbb{R}$ is the valuation of the player.  
\begin{definition}[\cite{DaskalakisSyrgkanis16:Learning-in-auctions:}]
An algorithm for the online learning problem is \emph{$r$-approximate no-envy} if, for any adaptively chosen sequence of (random) threshold 
vector $\theta^{1:T}$ by the adversary, the (random) bid vector $b^{1:T}$ chosen by the algorithm satisfies:
\begin{align}	\label{eq:no-envy}
 \frac{1}{T} \sum_{t=1}^{T} \E \bigl[ u(b^{t},\theta^{t}) \bigr] 
	\geq \max_{S \subset [m]} \biggl( \frac{1}{r} \cdot v(S) - \sum_{j \in S} \frac{1}{T} \sum_{t=1}^{T} \E \bigl[ \theta^{t}_{j} \bigr] \biggr) - \epsilon(T)
\end{align} 
where the \emph{no-envy rate} $\epsilon(T) \rightarrow 0$ while $T \rightarrow \infty$.
An algorithm is \emph{no-envy} if it is 1-approximate no-envy.
\end{definition}

Now we show the connection between primal-dual and no-envy learning
by revisiting the following theorem. As we will see, the notion of no-envy learning corresponds exactly 
to a constraint of a dual program. 

\begin{theorem}[\cite{DaskalakisSyrgkanis16:Learning-in-auctions:}]		\label{thm:pd-no-envy}
If $n$ players in a S2A use an $r$-approximate no-envy learning algorithm with envy rate $\epsilon(T)$
then in $T$ steps, the average welfare is at least $\frac{1}{2r}\textsc{Opt} - n \cdot \epsilon(T)$ 
where $\textsc{Opt}$ is the expected optimal welfare. 
\end{theorem}
\begin{proof}
Let $b^{t}_{i}$ be the bid vector of player $i$ where $b^{t}_{ij}$ is the bid of player $i$ on item $j$ in step $t$. 
In a S2A the threshold $\theta^{t}_{ij} = \max_{k \neq i} b^{t}_{kj}$. Consider the same primal and dual LPs in 
Section~\ref{sec:pd-non-smooth}. 

\paragraph{Dual variables.} 
Recall that $r$ is the approximation factor and $\epsilon(T)$ the no-envy rate of the learning algorithm. 
Define dual variables (similar to the ones in Section~\ref{sec:pd-non-smooth}) as follows.
\begin{align*}
\alpha_{i}(v_{i}) &:= r \cdot f_{i}(v_{i}) \cdot \E_{\vect{v}_{-i} \sim \vect{F}_{i}} 
	\biggl[ \frac{1}{T} \sum_{t=1}^{T} \E_{\vect{b}^{t}(v_{i},\vect{v}_{-i})} \bigl[ u_{i}\bigl(b^{t}_{i}, \theta^{t}_{i} \bigr) \bigr] \biggr] + r \cdot \epsilon(T)\\
\gamma_{i,j}(v_{i}) &:= r \cdot  f_{i}(v_{i}) \cdot \E_{\vect{v}_{-i} \sim \vect{F}_{i}} 
	\biggl[ \frac{1}{T} \sum_{t=1}^{T} \E_{\vect{b}^{t}(v_{i},\vect{v}_{-i})} \bigl[ \theta^{t}_{ij} \bigr] \biggr]
	= r \cdot f_{i}(v_{i}) \cdot \E_{\vect{v}_{-i} \sim \vect{F}_{i}} 
	\biggl[ \frac{1}{T} \sum_{t=1}^{T} \E_{\vect{b}^{t}_{-i}(\vect{v}_{-i})} \bigl[ \theta^{t}_{ij} \bigr] \biggr] \\
\beta_{j} &:=  r \cdot \max_{i} \max_{v_{i}} \E_{\vect{v}_{-i} \sim \vect{F}_{i}} 
	\biggl[ \frac{1}{T} \sum_{t=1}^{T} \E_{\vect{b}^{t}(v_{i},\vect{v}_{-i})} \bigl[ \theta^{t}_{ij} \bigr] \biggr]
	= r \cdot \max_{i} \E_{\vect{v}_{-i} \sim \vect{F}_{i}} 
	\biggl[ \frac{1}{T} \sum_{t=1}^{T} \E_{\vect{b}^{t}_{-i}(\vect{v}_{-i})} \bigl[ \theta^{t}_{ij} \bigr] \biggr]
\end{align*}
where the second equalities in the definitions of $\gamma$ and $\beta$ follow the fact that 
player valuations are independent and $\theta^{t}_{ij}$ does not depend on $b^{t}_{ij}$ for every $i,j$.  

\paragraph{Feasibility.} The first dual constraint follows immediately by the definitions of dual variables $\beta$ and $\gamma$. 
For a fixed set $S$ and a player $i$ with valuation $v_{i}$, the second dual constraint reads
\begin{align*}
 r \cdot f_{i}(v_{i}) \cdot \E_{\vect{v}_{-i} \sim \vect{F}_{i}} 
	&\biggl[ \frac{1}{T} \sum_{t=1}^{T} \E_{\vect{b}^{t}(v_{i},\vect{v}_{-i})} \bigl[ u_{i}\bigl(b^{t}_{i}, \theta^{t}_{i} \bigr) \bigr] \biggr]
+
r \cdot \epsilon(T) \\
&+
r \cdot \sum_{j \in S} f_{i}(v_{i}) \cdot \E_{\vect{v}_{-i} \sim \vect{F}_{i}} 
	 \biggl[ \frac{1}{T} \sum_{t=1}^{T} \E_{\vect{b}^{t}_{-i}(\vect{v}_{-i})} \bigl[ \theta^{t}_{ij} \bigr] \biggr] 
\geq f_{i}(v_{i}) \cdot v_{iS} 
\end{align*}
This inequality follows immediately from the definition of $r$-approximate no-envy learning algorithms (specifically, Inequality (\ref{eq:no-envy})).

\paragraph{Bounding the cost.}
In $T$ steps, the average welfare is 
$
\E_{\vect{v}} \bigl[ \frac{1}{T} \sum_{t=1}^{T} \E_{\vect{b}^{t}(\vect{v})} \bigl[ u_{i}\bigl(b^{t}_{i}, \theta^{t}_{i} \bigr) \bigr] \bigr] 
$.
Besides,
\begin{align*}
\sum_{i,v_{i}} \alpha_{i}(v_{i}) 
&\leq  r \cdot \E_{\vect{v}} \biggl[ \frac{1}{T} \sum_{t=1}^{T} \E_{\vect{b}^{t}(\vect{v})} \bigl[ v_{i}\bigl(S^{*}(b^{t}_{i}, \theta^{t}_{i}) \bigr) \bigr] \biggr] 
+ n \cdot r \cdot \epsilon(T), \\
\sum_{j} \beta_{j} 
&\leq r \cdot \E_{\vect{v}} \biggl[ \frac{1}{T} \sum_{t=1}^{T} \E_{\vect{b}^{t}(\vect{v})} \bigl[ v_{i}\bigl(S^{*}(b^{t}_{i}, \theta^{t}_{i}) \bigr) \bigr] \biggr] 
\end{align*}
where the last inequality is due to the non-overbidding property. Hence, the theorem follows by weak duality.
\end{proof}

%Note that, \citet{DaskalakisSyrgkanis16:Learning-in-auctions:} also
%proved the existence of no-envy learning algorithms in S2A when the players have XOS valuations.
%Therefore, as a corollary of Theorem~\ref{thm:pd-no-envy}, the PoA in S2A
%is at most 2 when player valuations are XOS.

\subsection{Sequential Auctions}

\subsubsection{Sequential Second Price Auctions in Sponsored Search}	\label{sec:sponsor}

\paragraph{Model.} 
In the sponsored search problem, there are $n$ players and $n$ slots. Each player $i$ has a \emph{private valuation}
$v_{i}$, representing its valuation per click.  We use $\vect{v} = (v_{1}, \ldots, v_{n})$ to denote
the \emph{valuation profile} of players. Additionally, each player $i$ has a \emph{quality factor} $\alpha_{i}$ 
that reflect the click-ability of the ad. The couple of valuation and quality factor 
$(v_{i},\alpha_{i})$ of player $i$ is drawn from a publicly known distribution $F_{i}$.
In the model, we assume that the distributions $F_{i}$'s are mutually independent. 
The  slots have associated \emph{click-through-rates} $\beta_{1} \geq \beta_{2} \geq \ldots \geq \beta_{n}$.
An \emph{outcome} is an one-to-one assignment of slots to players. 
When player $i$ is assigned to the $j$-th slot, the player gets $\alpha_{i}\beta_{j}$ clicks.

In the auction, the auctioneer sells slots sequentially one-by-one in non-increasing order of $\beta_{j}$
via the second price mechanisms. At the consideration of slot $j$, the auctioneer collects 
all the bid $b_{ij}$ on item $j$ from every player $i$ , 
which is interpreted as a valuation declaration. We also assume that the non-overbidding property, meaning that 
$b_{ij} \leq v_{i}$ for all $i$ and $j$. 
The auctioneer then assigns slot $j$ to the player (that has not received any slot so far)
with highest \emph{effective} bid, defined as $\alpha_{i} b_{i}$. 
The payment of the winning player is set according to \emph{critical} value: the smallest bid
that guarantees the player still gets the slot. Specifically, if a slot $j$ is assigned to player $i$ then the payment of $i$ is 
$p_{i} = \alpha_{i'}\beta_{i'}/\alpha_{i}$ where $\alpha_{i'}\beta_{i'}$ is the second highest effective bid on slot $j$.
The \emph{utility} of player $i$ is 
$\alpha_{i}\beta_{j} (v_{i} - p_{i})$. The \emph{social welfare} of the outcome is 
$\sum_{i,j}  \beta_{j} \alpha_{i} v_{i}$ where the sum is taken over all player $i$ with their allocated slots $j$.

This setting is captured by extensive form games (see \cite{FudenbergTirole91:Game-Theory,Peters15:Game-theory:} 
for comprehensive treatments).
The strategy of each player is an adaptive bidding policy: the bid of player $i$ for slot $j$ is a function of 
its valuation $v_{i}$, the common knowledge about the distributions of player valuations $\vect{F}$
and the history $h_{j}$ of outcomes in auctions before the consideration of slot $j$. Thus 
a player strategy can be denoted as $b_{ij}(v_{i},h_{j})$. We are interested in the \emph{perfect Bayesian equilibria}
which is a refinement of the concepts of Bayes-Nash equilibria and subgame perfect equilibria. 
A profile of bidding polices is a \emph{perfect Bayesian equilibrium} if it is a Bayes-Nash equilibrium of the original game
and given an arbitrary history (of some $t$ first rounds), the policy profile remains also a Bayes-Nash equilibrium of this induced game.

The sponsored search problem via the generalized second-price auctions has been extensively studied, 
first considered by \citet{MehtaSaberi07:Adwords-and-generalized} 
from optimization perspective and was proposed simultaneously by \citet{EdelmanOstrovsky07:Internet-advertising} 
and \citet{Varian07:Position-auctions}
from game theoretical viewpoint (see \cite{LahaiePennock07:Sponsored-search,MailleMarkakis12:Sponsored-search} for surveys on the topic). Recently, \citet{CaragiannisKaklamanis15:Bounding-the-inefficiency} have proved the PoA 
bound of 2.927 (without the independence assumption on distributions $F_{i}$'s), 
the currently best known PoA bound, using a technique called semi-smoothness, 
an extension of the smoothness framework in \cite{Roughgarden15:Intrinsic-robustness}. 
The study of PoA of sequential auctions in algorithmic game theory
has been initiated by \citet{LemeSyrgkanis12:Sequential-auctions}. 
The authors studied sequential first price auctions for matching markets and matroid auctions in 
the full-information environments and showed that the PoA (of pure Nash equilibria) is at most 2. 
Subsequently, \citet{SyrgkanisTardos12:Bayesian-sequential} extended the results to incomplete-informations settings 
and gave constant bounds for both auctions. \citet{LemeSyrgkanis12:Sequential-auctions,SyrgkanisTardos12:Bayesian-sequential}
proposed a bluffing deviation, where a player pretends to play as in equilibrium, until the right moment when 
the player deviates to acquire some item. This hypothetical deviation gives rise to useful inequalities to bound the PoA.

In this section, we show an improved bound of 2 over the best-known PoA bound of 
2.927~\cite{CaragiannisKaklamanis15:Bounding-the-inefficiency}.
In the analysis, the dual variables are intuitively constructed such that they correspond to the player utilities and 
player payments. In order to show the feasibility of dual variables, we also use the idea of bluffing deviations.
These deviations, coupling with the assumption of equilibrium, lead to useful inequalities which are served to prove the feasibility. 
The primal-dual approach indeed enables the improvement as well as a fairly simple proof.

\paragraph{Formulation.} 
For player $i$ with valuation $v_{i}$ and quality factor $\alpha_{i}$, 
let $x_{ij}(v_{i},\alpha_{i})$ be a variable indicating the interim assignment of slot $j$ to player $i$. 
Recall that $F_{i}$ is the distribution of $(v_{i},\alpha_{i})$.
Consider the following relaxation of the sponsored search problem and its dual.
In the primal relaxation, the first constraint says that a player receives at most one slot and 
the second one ensures that one slot is assigned to at most one player. 
%(The constraints can be constructed in the same way as the formulation of the previous section) 

\begin{minipage}[t]{0.4\textwidth}
\begin{align*}
\max ~ \sum_{i,j} \E_{(v_{i},\alpha_{i}) \sim F_{i}} \biggl[  \beta_{j} \alpha_{i} 
				v_{i} \cdot &x_{ij}(v_{i},\alpha_{i}) \biggr] \\
\sum_{j} x_{ij}(v_{i},\alpha_{i}) &\leq 1 \qquad  \forall i, v_{i}, \alpha_{i} \\
\sum_{i} \sum_{(v_{i},\alpha_{i})} f_{i}(v_{i},\alpha_{i}) x_{ij}(v_{i},\alpha_{i}) &\leq 1 \qquad  \forall j \\
x_{ij}(v_{i},\alpha_{i}) &\geq 0 \qquad \forall i,j,v_{i},\alpha_{i}
\end{align*}
\end{minipage}
\qquad
\begin{minipage}[t]{0.45\textwidth}
\begin{align*}
\min ~ \sum_{i} \sum_{(v_{i},\alpha_{i})} y_{i}(v_{i},&\alpha_{i}) + \sum_{j} z_{j}\\
		y_{i}(v_{i},\alpha_{i}) + f_{i}(v_{i},\alpha_{i}) z_{j} &\geq f_{i}(v_{i},\alpha_{i}) \cdot \beta_{j} \alpha_{i} v_{i} 
			\\ & \qquad \qquad \forall i,j,v_{i},\alpha_{i}\\
y_{i}(v_{i},\alpha_{i}), z_{j} &\geq 0 \qquad \forall i,j,v_{i},\alpha_{i} \\
\end{align*}
\end{minipage}

\begin{theorem}
For every sequential  second-price auction setting, 
the expected welfare of every perfect Bayesian equilibrium is at 
least half the maximum welfare.
\end{theorem}
\begin{proof}
Fix a Bayes-Nash equilibrium $\vect{\sigma}$. 
Let $\pi(\vect{\sigma}(\vect{v}, \vect{\alpha}), i)$ be the random variable indicating the 
slot that player $i$ receives in the equilibrium $\vect{\sigma}(\vect{v}, \vect{\alpha})$ given the valuation profile $\vect{v}$
and the quality factor profile $\vect{\alpha}$.
Whenever $\vect{\sigma}$ and $(\vect{v}, \vect{\alpha})$ are clear in the context, we simply write 
$\pi(\vect{\sigma}(\vect{v},\vect{\alpha}), i)$
as $\pi(i)$. Inversely, let $\pi^{-1}(\vect{\sigma}(\vect{v}, \vect{\alpha}),j)$ be the winner of slot $j$
in profile $\vect{\sigma}(\vect{v}, \vect{\alpha})$. 
Note that $\pi^{-1}(\vect{\sigma}(\vect{v}, \vect{\alpha}),j)$ is also a random variable.

\paragraph{Dual Variables.} For fixed $(v_{i},\alpha_{i})$, denote $\vect{B}(v_{i},\alpha_{i})$ the distribution of 
the equilibrium bid $\vect{b} = \linebreak \vect{\sigma}\bigl( (v_{i},\vect{v}_{-i}), (\alpha_{i}, \vect{\alpha}_{-i}) \bigr)$.
Recall that $\vect{b} = (b_{1}, \ldots, b_{n})$ where $b_{i}$ is a bid vector over bids $b_{ij}$ --- the equilibrium bid
that player $i$ submits in the round selling slot $j$. 
Moreover, denote $\vect{B}_{-i}$ the distribution of 
the equilibrium bid $\vect{b}_{-i} = \vect{\sigma}_{-i}\bigl( (v_{i},\vect{v}_{-i}), (\alpha_{i}, \vect{\alpha}_{-i}) \bigr)
= \vect{\sigma}_{-i}\bigl(\vect{v}_{-i}, \vect{\alpha}_{-i}) \bigr)$ where the last equality is due to the independence
of distributions $F_{i}$'s.
Define the dual variables as follows.
\begin{align*}
y_{i}(v_{i},\alpha_{i}) 
	&:= f_{i}(v_{i},\alpha_{i}) \cdot \E_{\vect{b} \sim \vect{B}(v_{i},\alpha_{i})} 
			\left[ \beta_{\pi(\vect{b}, i)} \cdot \alpha_{i}v_{i} \right], \\
z_{j}	&:= \max_{i} \E_{\vect{b}_{-i} \sim \vect{B}_{-i}} \left[ \beta_{j} \cdot \alpha_{\pi^{-1}(\vect{b}_{-i},j)} b_{\pi^{-1}(\vect{b}_{-i},j),j} \right] 
\end{align*}
Note that $\pi^{-1}(\vect{b}_{-i},j)$ is the winner of slot $j$ in the round selling slot $j$
assuming that player $i$ do not participate to this round. 

\paragraph{Feasibility.} Fix a player $i$ with valuation $v_{i}$ and quality factor $\alpha_{i}$, and a slot $j$. 
We show that the dual constraint corresponding to $i,j,v_{i},\alpha_{i}$ is satisfied. 
By the dual variable definitions and the independence of distributions, it is equivalent to prove that:
\begin{align}	\label{eq:GSP-dual}
\E_{\vect{b} \sim \vect{B}(v_{i},\alpha_{i})}
	\left[ \beta_{\pi(\vect{b}, i)} \cdot \alpha_{i}v_{i} + \beta_{j} \cdot \alpha_{\pi^{-1}(\vect{b}_{-i},j)} b_{\pi^{-1}(\vect{b}_{-i},j),j} \right] 
	\geq \beta_{j} \cdot \alpha_{i} v_{i}
\end{align}

We prove this inequality through a choice of a hypothetical deviation of player $i$ and use the assumption that 
$\vect{\sigma}$ is a Bayes-Nash equilibrium. We first make some observations. Consider a fixed valuation profile $\vect{v}_{-i}$,
a fixed quality factor profile $\vect{\alpha}_{-i}$
and a realization of (mixed) equilibrium $\vect{\sigma}\bigl( (v_{i},\vect{v}_{i}),(\alpha_{i},\vect{\alpha}_{-i}) \bigr)$, 
denoted as $\vect{b} = (b_{1}, \ldots, b_{n})$. 
Now the assignment $\pi$ of slots to players is completely determined. There are three different cases.
\begin{description}
	\item[Case 1:] Player $i$ receives some slot $\pi(i) \leq j$.
	 Then $ \beta_{\pi(i)} \cdot \alpha_{i}v_{i} \geq \beta_{j} \cdot \alpha_{i} v_{i}$
	 since $\beta_{\pi(i)} \geq \beta_{j}$.
	\item[Case 2:] $\pi(i) > j$ and $\alpha_{\pi^{-1}(j)}  b_{\pi^{-1}(j),j} \geq \alpha_{i} v_{i}$.
	Then $\beta_{j} \cdot \alpha_{\pi^{-1}(j)} \cdot b_{\pi^{-1}(j)} \geq \beta_{j} \cdot \alpha_{i} v_{i}$. 
	\item[Case 3:] $\pi(i) > j$ and $\alpha_{\pi^{-1}(j)}  b_{\pi^{-1}(j),j} < \alpha_{i} v_{i}$.
	Note that in this case in the round $j$, player $i$ could have submitted a bid without violating the no-overbidding property 
	such that the corresponding effective bid is infinitesimal larger than $\alpha_{\pi^{-1}(j)} b_{\pi^{-1}(j),j}$
	and could have received slot $j$. 
\end{description}

We are now choosing a bid deviation in order to prove the dual constraint based on the fact that
$\vect{\sigma}$ is a Bayes-Nash equilibrium. Intuitively, the different cases above suggest the following deviation.
For the first two cases, the term inside the expectations in the left-hand-side of (\ref{eq:GSP-dual})
is already larger than the right-hand-side (so no need to deviate). Hence, the deviation is necessary only in Case 3.  

Formally, we define the (mixed) deviation $b^{'}_{i}$ as follows. 
First, player $i$ follows the equilibrium strategy $b_{i}$. If until the allocation step of slot $j$, player $i$
has not received to any slot then submit $v_{i}$. 
As $\vect{\sigma}$ is a Bayes-Nash equilibrium, 
the utility of player $i$ is at least that induced by this deviation. Specifically,
\begin{align*}
\E_{\vect{b} \sim \vect{B}(v_{i},\alpha_{i})} \bigl[ u_{i}(\vect{b}) \bigr]
\geq 
\E_{\vect{b}_{-i} \sim \vect{B}_{-i}} \E_{b'_{i}} 
		 \biggl[ u_{i}( b'_{i} ,\vect{b}_{-i}) \biggr]
\end{align*}
where since $(v_{i},\alpha_{i})$ is fixed, for short, we write $u_{i}(\vect{b}) = u_{i}(\vect{b}; v_{i},\alpha_{i})$.

By definition of the deviation $b'_{i}$, player $i$ follows the same equilibrium strategy $b_{i}$ if Case 1 happens. 
Therefore, the above inequality is equivalent to 
\begin{align}	\label{ineq:sponsor-deviation}
\E_{\vect{b} \sim \vect{B}(v_{i},\alpha_{i})} 
		\biggl[ u_{i}(\vect{b}) \bigr | \textnormal{Case 2 or Case 3} \biggr]
\geq
\E_{\vect{b}_{-i} \sim \vect{B}_{-i}} \E_{b'_{i}} 
		 \biggl[ u_{i}( b'_{i} ,\vect{b}_{-i}) \bigr | \textnormal{Case 2 or Case 3} \biggr]
\end{align}

Note that if Case 3 holds then player $i$ gets slot $j$ with the payment 
$\frac{\alpha_{\pi^{-1}(\vect{b}_{-i},j)} b_{\pi^{-1}(\vect{b}_{-i},j),j}}{\alpha_{i}}$. 
So

\begin{align}	\label{ineq:sponsor-deviation-refine}
\E_{\vect{b}_{-i} \sim \vect{B}_{-i}} \E_{b'_{i}} 
		 \biggl[ u_{i}( b'_{i} ,\vect{b}_{-i}) \bigr | \textnormal{Case 3} \biggr]
&=
\E_{\vect{b}_{-i} \sim \vect{B}_{-i}}
	\left[  \beta_{j} \cdot \alpha_{i} \biggl( v_{i} - \frac{\alpha_{\pi^{-1}(\vect{b}_{-i},j)} b_{\pi^{-1}(\vect{b}_{-i},j),j}}{\alpha_{i}} \biggr) \big | \textnormal{Case 3}  \right] 
\end{align}

We are now ready to prove the inequality (\ref{eq:GSP-dual}). We have
\begin{align*}
& \E_{\vect{b} \sim \vect{B}(v_{i},\alpha_{i})} 
	\left[ \beta_{\pi(\vect{b}, i)} \cdot \alpha_{i}v_{i} + \beta_{j} \cdot \alpha_{\pi^{-1}(\vect{b}_{-i},j)} b_{\pi^{-1}(\vect{b}_{-i},j),j} \right] \\
&= \sum_{\ell = 1, 2,3 } \E_{\vect{b} \sim \vect{B}(v_{i},\alpha_{i})}
	\left[ \beta_{\pi(\vect{b}, i)} \cdot \alpha_{i}v_{i} + \beta_{j} \cdot \alpha_{\pi^{-1}(\vect{b}_{-i},j)} b_{\pi^{-1}(\vect{b}_{-i},j),j} \big| \textnormal{Case } \ell \right] \\
&\geq \E_{\vect{b} \sim \vect{B}(v_{i},\alpha_{i})} \left[ \beta_{j} \cdot \alpha_{i} v_{i} \big | \textnormal{Case 1} \right]
	+   \E_{\vect{b} \sim \vect{B}(v_{i},\alpha_{i})}  \left[ \beta_{\pi(\vect{b}, i)} \cdot \alpha_{i}v_{i} + \beta_{j} 
					\cdot \alpha_{\pi^{-1}(\vect{b}_{-i},j)} b_{\pi^{-1}(\vect{b}_{-i},j),j} \big| \textnormal{Case 2 or 3} \right] \\
&\geq \E_{\vect{b} \sim \vect{B}(v_{i},\alpha_{i})} \left[ \beta_{j} \cdot \alpha_{i} v_{i} \big | \textnormal{Case 1} \right]
	+   \E_{\vect{b} \sim \vect{B}(v_{i},\alpha_{i})} \left[ u_{i}(\vect{b}) + \beta_{j} 
					\cdot \alpha_{\pi^{-1}(\vect{b}_{-i},j)} b_{\pi^{-1}(\vect{b}_{-i},j),j} \big| \textnormal{Case 2 or 3} \right] \\
&\geq \E_{\vect{b} \sim \vect{B}(v_{i},\alpha_{i})} \left[ \beta_{j} \cdot \alpha_{i} v_{i} \big | \textnormal{Case 1} \right]
	+   \E_{\vect{b} \sim \vect{B}(v_{i},\alpha_{i})}  \left[ u_{i}\bigl((b'_{i},\vect{b}_{-i}) \bigr) + \beta_{j} 
					\cdot \alpha_{\pi^{-1}(\vect{b}_{-i},j)} b_{\pi^{-1}(\vect{b}_{-i},j),j} \big| \textnormal{Case 2 or 3} \right] \\
&\geq \E_{\vect{b} \sim \vect{B}(v_{i},\alpha_{i})} \left[ \beta_{j} \cdot \alpha_{i} v_{i} \big | \textnormal{Case 1 or 2} \right]
	+   \E_{\vect{b} \sim \vect{B}(v_{i},\alpha_{i})}  \left[ u_{i}\bigl((b'_{i},\vect{b}_{-i}) \bigr) 
				+ \beta_{j} \cdot \alpha_{\pi^{-1}(\vect{b}_{-i},j)} b_{\pi^{-1}(\vect{b}_{-i},j),j} \big| \textnormal{Case 3} \right] \\
&\geq \E_{\vect{b} \sim \vect{B}(v_{i},\alpha_{i})}\left[ \beta_{j} \cdot \alpha_{i} v_{i} \big | \textnormal{Case 1 or 2} \right] \\
  &\qquad	\qquad 
  	+  \E_{\vect{b} \sim \vect{B}(v_{i},\alpha_{i})}  \left[  \beta_{j} \cdot \alpha_{i} v_{i} 
				  - \beta_{j} \cdot \alpha_{\pi^{-1}(\vect{b}_{-i},j)} b_{\pi^{-1}(\vect{b}_{-i},j),j}
  				 + \beta_{j} \cdot \alpha_{\pi^{-1}(\vect{b}_{-i},j)} b_{\pi^{-1}(\vect{b}_{-i},j),j} \big| \textnormal{Case 3} \right] \\
&= \E_{\vect{b} \sim \vect{B}(v_{i},\alpha_{i})} \left[ \beta_{j} \cdot \alpha_{i} v_{i} \right] =  \beta_{j} \cdot \alpha_{i} v_{i}
\end{align*}
The first inequality follows the assumption of Case 1: $\beta_{\pi(i)} \geq \beta_{j}$. The second inequality holds since
the utility $u_{i}(\vect{b}) \leq \beta_{\pi(\vect{b},i)} \cdot \alpha_{i} v_{i}$. The third inequality is due to 
(\ref{ineq:sponsor-deviation}). The fourth inequality follows the assumption of Case 2: $\alpha_{\pi^{-1}(j)}  v_{\pi^{-1}(j)} \geq \alpha_{i} v_{i}$.
The last inequality follows (\ref{ineq:sponsor-deviation-refine}). Hence, the constructed dual variables form a dual feasible solution.

\paragraph{Bounding primal and dual.} 
Let $\vect{B}$ be the distribution of equilibrium bid $\vect{b} = \vect{\sigma}(\vect{v})$.
The expected welfare of equilibrium $\vect{\sigma}$ is 
$
\E_{\vect{b} \sim \vect{B}} \bigl[ \sum_{i}  \beta_{\pi(\vect{b},i)} \alpha_{i} v_{i} \bigr]
$.
By the definition of dual variables, we have 
\begin{align*}
\sum_{i,(v_{i}, \alpha_{i})} y_{i}(v_{i},\alpha_{i})
= \sum_{i}  \E_{(v_{i},\alpha_{i}) \sim F_{i}} \E_{\vect{b} \sim \vect{B}(v_{i},\alpha_{i})} \biggl[  \beta_{\pi(\vect{b},i)} \alpha_{i} v_{i} \biggr]
= \E_{\vect{b} \sim \vect{B}} \biggl[ \sum_{i}  \beta_{\pi(\vect{b},i)} \alpha_{i} v_{i} \biggr].
\end{align*}
Besides, consider a slot $j$ and let $i^{*}$ be the player such that 
$$
z_{j} = \E_{\vect{b}_{-i^{*}} \sim \vect{B}_{-i^{*}} }
			\left[ \beta_{j} \cdot \alpha_{\pi^{-1}(\vect{b}_{-i^{*}},j)} b_{\pi^{-1}(\vect{b}_{-i^{*}},j),j} \right]
$$ 
As the right-hand side is independent of $b_{i^{*}}$, we have 
$$
z_{j} = \E_{b_{i^{*}}} \E_{\vect{b}_{-i^{*}} \sim \vect{B}_{-i^{*}} }
			\left[ \beta_{j} \cdot \alpha_{\pi^{-1}(\vect{b}_{-i^{*}},j)} b_{\pi^{-1}(\vect{b}_{-i^{*}},j),j} \right]
       = \E_{\vect{b}}
			\left[ \beta_{j} \cdot \alpha_{\pi^{-1}(\vect{b}_{-i^{*}},j)} b_{\pi^{-1}(\vect{b}_{-i^{*}},j),j} \right]
$$
Moreover,
\begin{align*}
z_{j} &\leq  \E_{\vect{b}} \left[ \beta_{j} \cdot \alpha_{\pi^{-1}(\vect{b},j)} b_{\pi^{-1}(\vect{b},j),j} \right]
\leq  \E_{\vect{b}} \left[ \beta_{j} \cdot \alpha_{\pi^{-1}(\vect{b},j)} v_{\pi^{-1}(\vect{b},j)} \right]	
\end{align*}
The first inequality holds since the effective bid of the slot-$j$-winner in round $j$
including all players is larger than that in case player $i^{*}$ does not participate.
The last inequality is due to the non-overbidding property. Summing over all $j$, we have
\begin{align*}
\sum_{j}z_{j} &\leq \E_{\vect{b}}
			\biggl[ \sum_{j} \beta_{j} \cdot \alpha_{\pi^{-1}(\vect{b},j)} v_{\pi^{-1}(\vect{b},j)} \biggr]	
		= \E_{\vect{b}} \biggl[ \sum_{i}  \beta_{\pi(\vect{b},i)} \alpha_{i} v_{i} \biggr].
\end{align*}
Thus, the dual objective value is at most twice the expected welfare of the equilibrium. 
\end{proof}

\paragraph{Remark.} The non-overbidding assumption can be relaxed in the same way as the remark in 
Section~\ref{sec:pd-non-smooth}. Specifically, if players are allowed to bid up to a constant $r$ times their valuations 
(hence, the utility of a winning player may be negative) then the PoA is at most $(1+r)$.  

\subsubsection{Sequential First Price Auctions in Matching Markets}

\paragraph{Model.} 
In the matching market problem, there are $n$ players and $m$ items. Each player $i$ has \emph{private unit-demand valuation}
$v_{i}: 2^{[m]} \rightarrow \mathbb{R}$ defined as $v_{iS} := \max_{j} v_{ij}$ where $v_{ij}$ is the valuation of player $i$ on item $j$.
Note that in the sponsored search problem $v_{ij} \geq v_{ij'}$ for every $j < j'$ and for every player $i$, 
while in the matching market problem it might be that for some items $j,j'$ and some players $i, i'$, 
$v_{ij} > v_{ij'}$ and $v_{i'j} < v_{i'j'}$.   
The valuation vector $v_{i}$ is drawn from a publicly known distribution $F_{i}$.
In the model, we assume that the distributions $F_{i}$'s are mutually independent. 
An \emph{outcome} is an assignment of items to players.

In the auction, the auctioneer sells items sequentially one-by-one
via the first price mechanisms. At the consideration of item $j$, the auctioneer collects 
all the bids $b_{ij}$ on item $j$ from all players.
We also assume that the non-overbidding property, meaning that $b_{ij} \leq v_{i}$ for all $i$ and $j$. 
The auctioneer then assigns item $j$ to the player with highest bid. 
Note that, in contrast to the sponsored search problem, a player may receive multiple items. 
The payment of the winning player is simply the winning bid.
The \emph{utility} of player $i$ is $(v_{iS} - \sum_{j \in S} b_{ij})$ where $S$ is its allocated items. 
The \emph{social welfare} of the outcome is 
$\sum_{i,j} v_{iS}$ where the sum is taken over all players $i$ and their corresponding allocated items $S$.

Related work about sequential auctions have been summarized in the previous section.
For the matching market problem, \citet{LemeSyrgkanis12:Sequential-auctions} proved that the sequential auctions via the second price mechanisms
may lead to unbounded inefficiency.  The authors \cite{LemeSyrgkanis12:Sequential-auctions} then 
considered the sequential first price auctions and showed that 
in full-information settings, the PoA is at most 2 and 4 for pure and mixed Nash equilibria. 
Subsequently, \citet{SyrgkanisTardos12:Bayesian-sequential} extended the results to incomplete-information settings.
They proved a Bayesian PoA bound $2e/(e-1)$ for matching markets with independent valuations. 
They also raised a question whether the difference of PoA bounds between the full-information settings and the incomplete-information 
ones is necessary. 

In this section, we answer this question by showing that the (mixed) Bayesian PoA is at most 2. 
In the proof, we use similar bluffing deviations as in \cite{LemeSyrgkanis12:Sequential-auctions,SyrgkanisTardos12:Bayesian-sequential} 
and the primal-dual approach enables the improvement. 
The proof follows similar structure as the one in Section~\ref{sec:sponsor}; however, there is a subtle difference compared to 
the sponsored search problem. In the latter, each player receives at most one item (slot) so in constructing the hypothetical deviation, 
it is sufficient to design a deviation in which the player gets one item, improves its utility and then leaves the game (bids 0 in subsequent rounds). 
In the matching market problem, a player may receive multiple 
items hence the player would deviate in such a way that the player receives only the highest valuable item without receiving (so paying for) items 
allocated in previous rounds.
However, such deviations may lead to completely different outcomes and the equilibrium structure could be very complex to analyze. 
Therefore, we do not reason directly on the utility of players in deviation. Instead, we explore the connection between the winning bid 
and the player valuation. Consequently, the argument works only for the sequential auctions via the first price mechanisms 
(but not via the second price mechanisms).

\paragraph{Formulation.} 
For every player $i$, every valuation $v_{i}$ and every set of items $S$, 
let $x_{iS}(v_{i})$ be a variable indicating the interim assignment of $S$  to player $i$. 
Consider the following formulation and its dual.
In the primal, the first and second constraints are relaxations of the facts that 
a player receives a set of items and 
an item is assigned to at most one player, respectively.

\begin{minipage}[t]{0.4\textwidth}
\begin{align*}
\max ~ \sum_{i,S} \E_{v_{i} \sim F_{i}} \bigl[ v_{iS} \cdot &x_{iS}(v_{i}) \bigr] \\
\sum_{S} x_{iS}(v_{i}) &\leq 1 \qquad  \forall i, v_{i} \\
\sum_{i} \sum_{v_{i}} f_{i}(v_{i}) \sum_{S: j \in S}x_{iS}(v_{i}) &\leq 1 \qquad  \forall j \\
x_{iS}(v_{i}) &\geq 0 \qquad \forall i,j,v_{i}
\end{align*}
\end{minipage}
\qquad
\begin{minipage}[t]{0.45\textwidth}
\begin{align*}
\min ~ \sum_{i} \sum_{v_{i}} y_{i}(v_{i}) &+ \sum_{j} z_{j}\\
		y_{i}(v_{i}) + f_{i}(v_{i}) \sum_{j \in S} z_{j} &\geq f_{i}(v_{i}) \cdot v_{iS} 
				 & & \forall i,S,v_{i}\\
y_{i}(v_{i}), z_{j} &\geq 0 & & \forall i,j,v_{i} \\
\end{align*}
\end{minipage}

\begin{theorem}
For every sequential first-price auction, 
the expected welfare of every perfect Bayesian equilibrium is at 
least half the maximum welfare.
\end{theorem}
\begin{proof}
Fix a Bayes-Nash equilibrium $\vect{\sigma}$. 
Let $\pi(\vect{\sigma}(\vect{v}), i)$ be the random variable indicating the 
set of items allocated to player $i$ in the equilibrium given the valuation profile $\vect{v}$.
%Whenever $\vect{\sigma}$ and $\vect{v}$ are clear in the context, we write simply 
%$\pi(\vect{\sigma}(\vect{v}), i)$ as $\pi(i)$. 
Inversely, let $\pi^{-1}(\vect{\sigma}(\vect{v}),j)$ be the winner of item $j$. 
Note that $\pi^{-1}(\vect{\sigma}(\vect{v}),j)$ is also a random variable.

\paragraph{Dual Variables.} For a fixed valuation $v_{i}$, denote $\vect{B}(v_{i})$ the distribution of 
the equilibrium bid $\vect{b} = \vect{\sigma}(v_{i},\vect{v}_{-i})$.
Recall that $\vect{b} = (b_{1}, \ldots, b_{n})$ where $b_{i}$ is a bid vector over $b_{ij}$ --- the equilibrium bid
that player $i$ submits in the round selling item $j$ for $1 \leq j \leq m$. 
Moreover, denote $\vect{B}_{-i}$ the distribution of 
the equilibrium bid $\vect{b}_{-i} = \vect{\sigma}_{-i}(v_{i},\vect{v}_{-i})
= \vect{\sigma}_{-i}(\vect{v}_{-i})$ where the last equality is due to the independence
of distributions.
Define the dual variables as follows.
\begin{align*}
y_{i}(v_{i}) 
	&:= f_{i}(v_{i}) \cdot \E_{\vect{b} \sim \vect{B}(v_{i})} 
			\left[ v_{i,\pi(\vect{b},i)} \right], \\
z_{j}	&:= \max_{i} \E_{\vect{b}_{-i} \sim \vect{B}_{-i}} \left[ b_{\pi^{-1}(\vect{b}_{-i},j),j} \right] 
\end{align*}
Note that $\pi^{-1}(\vect{b}_{-i},j)$ is the winner of item $j$
assuming that player $i$ does not participate to this round. 

\paragraph{Feasibility.} Fix a player $i$ with valuation $v_{i}$ and a set of items $S$. 
We show that the dual constraint corresponding to $i,S,v_{i}$ is satisfied. 
By the dual variable definitions and the independence of distributions, it is equivalent to prove that:
\begin{align}	\label{eq:matching-dual}
\E_{\vect{b} \sim \vect{B}(v_{i})}
	\biggl[ v_{i,\pi(\vect{b},i)} + \sum_{j \in S}b_{\pi^{-1}(\vect{b}_{-i},j),j} \biggr] 
	\geq v_{iS}
\end{align}

We prove this inequality through a choice of a hypothetical deviation of player $i$ and use the assumption that 
$\vect{\sigma}$ is a Nash-Bayes equilibrium. 
For any set of items $U$, let $j^{*}(U) \in U$ be an item such that $v_{j^{*}} = \max_{j \in U} v_{ij} = v_{iU}$.
We first make some observations. Consider a fixed valuation profile $\vect{v}_{-i}$
and a realization of (mixed) equilibrium $\vect{\sigma}(v_{i},\vect{v}_{i})$, 
denoted as $\vect{b} = (b_{1}, \ldots, b_{n})$. 
Now the assignment $\pi$ of items to players is completely determined. 
Let $T = \pi(\vect{b},i)$.
There are three different cases.
\begin{description}
	\item[Case 1:] $v_{i,j^{*}(T)} \geq v_{i,j^{*}(S)}$. 
	\item[Case 2:] $v_{i,j^{*}(T)} < v_{i,j^{*}(S)}$ (so $j^{*}(S) \notin T$) and the round of $j^{*}(S)$ is before the round of $j^{*}(T)$. 
	In this case, $b_{\pi^{-1}(j^{*}(S)),j^{*}(S)} \geq v_{i,j^{*}(S)} - v_{i,j^{*}(T)}$ since otherwise $i$ could have improved 
	its utility by submitting a bid of value $(v_{i,j^{*}(S)} - v_{i,j^{*}(T)})$ and stop playing the remaining rounds (by submitting bids 0). 
	\item[Case 3:] $v_{i,j^{*}(T)} < v_{i,j^{*}(S)}$ (so $j^{*}(S) \notin T$) and the round of $j^{*}(T)$ is before the round of $j^{*}(S)$.
	Again, in this case, $b_{\pi^{-1}(j^{*}(S)),j^{*}(S)} \geq v_{i,j^{*}(S)} - v_{i,j^{*}(T)}$ by the same argument. 
\end{description}

The cases suggest the following (mixed) deviation $b^{'}_{i}$ of player $i$. 
Player $i$ draws a random sample of a valuation profile $\vect{w}_{-i} \in \vect{F}_{-i}$
and determine the winning set $T = \pi(\vect{\sigma}(v_{i},\vect{w}_{-i}),i)$ and also item $j^{*}(T)$. 
If $v_{i,j^{*}(T)} \geq v_{i,j^{*}(S)}$ then player $i$ follows the equilibrium strategy $b_{i}$. 
Otherwise, player $i$ first follows strategy $b_{i}$ until the round of item $j^{*}(S)$. In the round of $j^{*}(S)$,
bid $b'_{i,j^{*}(S)} = v_{i,j^{*}(S)} - v_{i,j^{*}(T)}$ and in the subsequent rounds, bid 0.

As $\vect{\sigma}$ is a Bayes-Nash equilibrium, 
the utility of player $i$ is at least that induced by this deviation. Specifically,
\begin{align*}
\E_{\vect{b} \sim \vect{B}(v_{i})} \bigl[ u_{i}(\vect{b}) \bigr]
= \E_{\vect{v}_{-i} \sim \vect{B}_{-i}} \E_{\vect{\sigma}} 
		 \biggl[ u_{i} \bigl( b_{i},\vect{\sigma}_{-i}(\vect{v}_{-i}) \bigr) \biggr]
\geq 
\E_{\vect{w}_{-i} \sim \vect{B}_{-i}} \E_{\vect{\sigma}} 
		 \biggl[ u_{i} \bigl( b'_{i},\vect{\sigma}_{-i}(\vect{w}_{-i}) \bigr) \biggr]
\end{align*}
where since $v_{i}$ is fixed, for short, we write $u_{i}(\vect{b}) = u_{i}(\vect{b}; v_{i})$.
By definition of the deviation $b'_{i}$, player $i$ follows the same equilibrium strategy $b_{i}$ if Case 1 happens. 
Therefore, by remaining variables, the above inequality implies 
\begin{align}	\label{ineq:matching-deviation}
\E_{\vect{v}_{-i} \sim \vect{B}_{-i}} \E_{\vect{\sigma}} 
		\biggl[ b_{\pi^{-1}\bigl( \vect{\sigma}_{-i}(\vect{v}_{-i}),j^{*}(S) \bigr),j^{*}(S )}  \big| \textnormal{Case 2 or 3} \biggr] 
\geq
\E_{\vect{v}_{-i} \sim \vect{B}_{-i}} \E_{\vect{\sigma}} 
		\biggl[ v_{i,j^{*}(S)} - v_{i,j^{*}(T)} \big| \textnormal{Case 2 or 3} \biggr]
\end{align}
where $T = \pi \bigl(\vect{\sigma}_{-i}(v_{i},\vect{v}_{-i}) ,i \bigr)$
the set of items allocated to $i$.

We are now ready to prove the inequality (\ref{eq:matching-dual}). We have
\begin{align*}
\E_{\vect{b} \sim \vect{B}(v_{i})} &
	\biggl[ v_{i,\pi(\vect{b},i)} + \sum_{j \in S}b_{\pi^{-1}(\vect{b}_{-i},j),j} \biggr]  \\
&\geq \sum_{\ell = 1, 2,3 } \E_{\vect{b} \sim \vect{B}(v_{i})}
	\biggl[ v_{i,\pi(\vect{b},i)} + b_{\pi^{-1}(\vect{b}_{-i},j^{*}(S)),j^{*}(S)} \big| \textnormal{Case } \ell \biggr] \\
&\geq \E_{\vect{b} \sim \vect{B}(v_{i})} \left[v_{iS} \big | \textnormal{Case 1} \right]
	+   \E_{\vect{b} \sim \vect{B}(v_{i})}  
		\left[ v_{i,\pi(\vect{b},i)} + b_{\pi^{-1}(\vect{b}_{-i},j^{*}(S)),j^{*}(S)}  \big| \textnormal{Case 2 or 3} \right] \\
&\geq \E_{\vect{b} \sim \vect{B}(v_{i})} \left[v_{iS} \big | \textnormal{Case 1} \right]
	+   \E_{\vect{b} \sim \vect{B}(v_{i})}  
		\left[ v_{i,\pi(\vect{b},i)} + \bigl( v_{i,j^{*}(S)} - v_{i,j^{*}(T)} \bigr) \big| \textnormal{Case 2 or 3} \right] \\
&= v_{iS}
\end{align*}
The first inequality holds since $j^{*}(S) \in S$ and the bids are non-negative. The second inequality holds due to the assumption
of Case 1. The third inequality follows Inequality (\ref{ineq:matching-deviation}). Hence, the constructed dual variables form a dual feasible solution.

\paragraph{Bounding primal and dual.} By the definition of dual variables, we have 
\begin{align*}
\sum_{i,v_{i}} y_{i}(v_{i})
= \sum_{i}  \E_{v_{i} \sim F_{i}} \E_{\vect{b} \sim \vect{B}(v_{i})} \biggl[ v_{i,\pi(\vect{b},i)} \biggr]
= \E_{\vect{b}} \biggl[ \sum_{i} v_{i,\pi(\vect{b},i)} \biggr].
\end{align*}
Besides, consider an item $j$ and let $i^{*}$ be player such that 
$$
z_{j} = \E_{\vect{b}_{-i^{*}} \sim \vect{B}_{-i^{*}} }
			\left[ b_{\pi^{-1}(\vect{b}_{-i^{*}},j),j} \right]
$$ 
As the right-hand side is independent of $b_{i^{*}}$, we have 
$$
z_{j} = \E_{b_{i^{*}}} \E_{\vect{b}_{-i^{*}} \sim \vect{B}_{-i^{*}} } \left[ b_{\pi^{-1}(\vect{b}_{-i^{*}},j),j} \right]
       = \E_{\vect{b}} \left[ b_{\pi^{-1}(\vect{b}_{-i^{*}},j),j} \right]
       \leq \E_{\vect{b}} \left[ b_{\pi^{-1}(\vect{b},j),j} \right]
$$
Summing over all items $j$, we get
\begin{align*}
\sum_{j}z_{j} &\leq \E_{\vect{b}} \biggl[ \sum_{j} b_{\pi^{-1}(\vect{b},j),j} \biggr]	
		= \E_{\vect{b}} \biggl[ \sum_{i} \sum_{j \in \pi(\vect{b},i)} b_{ij} \biggr]
		\leq  \E_{\vect{b}} \biggl[ \sum_{i} v_{i,\pi(\vect{b},i)} \biggr]
\end{align*}
where the last inequality is due to non-overbidding property. 
Thus, the dual objective value is at most twice the expected welfare of the equilibrium. 
\end{proof}

\section{Conclusion}
In the paper, we have presented a primal-dual approach to study the efficiency of games. 
We have shown the applicability of the approach on a wide variety of settings and gave simple
and improved analyses for several problems in settings of different natures. Beyond concrete results, 
the main point of the paper is to illuminate the potential of the primal-dual approach. In this approach, 
the PoA-bound analyses now can be done similarly as the analyses of LP-based algorithms 
in Approximation/Online Algorithms. We hope that 
linear programming and duality would bring new ideas and techniques, from well-developed domains 
such as approximation, online algorithms, etc to algorithmic game theory, not only for
the analyses and the understanding of current games but also for the design of new games (auctions) and 
new concepts leading to improved efficiency.    

\paragraph{Acknowledgement.} We thank Tim Roughgarden for pointing out some related works.

%\newpage
%\section{Directions}
%\begin{enumerate}
%	\item Coarse NE, best responses, regret moves, resource augmentation.
%	\item Bound on the POS for congestion games? 
%	\item How about the POS for network cost sharing? (seems hard) Reexplain the constant POS of 
%		broadcast connection games by duality.
%	\item Connections with smooth argument in the works of Roughgarden: incomplete-information,
%		large games, POA and POS for Shapley congestion games
%	\item Mechanism design, second price auction. Many work.
%	\item Opinion formation games
%	\item Scheduling games (with release times), extension for egalitarian objective.  
%	\item Routing over time, improving Kulkarni SODA games.
%	\item Strong NE
%	\item Connection Games
%	\item See if we can prove that Mayerson gives Bayes Nash that maximizes the revenue. Replace $v_{i}$ by $\phi_{i}$.
%\end{enumerate}

%\bibliographystyle{plainnat}
%\bibliography{game} 

\end{document}